\documentclass[
a4paper
]{article}

\usepackage[utf8]{inputenc}
\usepackage{amsmath}
\usepackage{amssymb}
\usepackage{amsthm}
\usepackage{mathtools}
\usepackage{fullpage}

\usepackage{xspace}
\usepackage{graphicx}
\usepackage{MnSymbol}
\usepackage[linesnumbered,ruled,vlined]{algorithm2e}
\usepackage{subcaption}
\usepackage{caption}
\usepackage{xfrac}

\usepackage{tikz}
\usetikzlibrary{automata,angles,calc,graphs,positioning,chains,arrows,decorations.pathmorphing}

\theoremstyle{plain}
\newtheorem{theorem}{Theorem}[section]
\newtheorem*{theorem*}{Theorem}
\newtheorem{proposition}[theorem]{Proposition}
\newtheorem*{proposition*}{Proposition}

\newtheorem*{lemma*}{Lemma}
\newtheorem{corollary}[theorem]{Corollary}
\newtheorem*{corollary*}{Corollary}

\newtheorem*{claim*}{Claim}

\theoremstyle{definition}

\newtheorem*{definition*}{Definition}

\theoremstyle{remark}

\newtheorem*{observation*}{Observation}
\newtheorem{example}[theorem]{Example}
\newtheorem*{example*}{Example}

\newtheorem{property}{Property}

\DeclareMathOperator*{\argmax}{arg\,max}

\newcommand{\bbN}{\mathbb{N}}

\newcommand{\bbQ}{\mathbb{Q}}
\newcommand{\bbQgeqz}{\bbQ_{\geq 0}}

\renewcommand{\paragraph}[1]{\smallskip
	\noindent\textbf{#1}.}

\newcommand{\sem}[1]{{\lsem{}{#1}\rsem}}

\newcommand{\U}{\mathcal{U}}
\newcommand{\setsU}{\finite(\U)}
\newcommand{\deltasum}{\delta_{\operatorname{sum}}}
\newcommand{\deltamin}{\delta_{\operatorname{min}}}
\newcommand{\deltaW}{\delta_{\operatorname{W}}}
\newcommand{\deltasummin}{\delta_{\operatorname{sum-min}}}

\newcommand{\dd}{\textbf{d}}
\newcommand{\ud}{\textbf{u}}
\newcommand{\udR}{\ud_{\operatorname{rel}}}

\newcommand{\Ball}{\mathcal{B}}
\newcommand{\tparent}{\operatorname{parent}}
\newcommand{\tree}{\mathcal{T}}
\newcommand{\utree}[1][S]{\tree_{#1}}
\newcommand{\tchildren}{\operatorname{children}}

\newcommand{\radius}{\operatorname{r}_S}
\newcommand{\uv}{V_{S}}
\newcommand{\ue}{E_{S}}

\newcommand{\cand}{\mathsf{C}}

\newcommand{\irI}{\mathcal{I}}

\newcommand{\irroot}{\texttt{Root}}
\newcommand{\irchildren}{\texttt{Children}}
\newcommand{\irmember}{\texttt{Member}}

\newcommand{\nop}[1]{}

\newcommand{\bin}{\text{bin}}

\newcommand{\diversityComputation}{\mathtt{DiversityComputation}}
\newcommand{\diversityExplicit}{\mathtt{DiversityExplicit}}
\newcommand{\diversityImplicit}{\mathtt{DiversityImplicit}}
\newcommand{\finite}{\operatorname{finite}}

\newcommand{\fdelta}{f_{\delta}}
\newcommand{\fQ}{f_{Q}}
\newcommand{\fT}{f_{\mathcal{T}}}

\title{Towards Tractability of the Diversity of Query Answers: Ultrametrics to the Rescue}

\usepackage{authblk}

\author[1]{Marcelo Arenas}
\author[2]{Timo Camillo Merkl}
\author[2]{Reinhard Pichler}
\author[1]{Cristian Riveros}
\affil[1]{Pontificia Universidad Católica de Chile, \texttt{marenas@uc.cl}, \texttt{cristian.riveros@uc.cl}}
\affil[2]{TU Wien, Austria, \texttt{timo.merkl@tuwien.ac.at}, \texttt{pichler@dbai.tuwien.ac.at}}

\date{} 
 
\begin{document}

\maketitle

\begin{abstract}
The set of answers to a query may be very large, potentially
overwhelming users when presented with the entire set. In such cases,
presenting only a small subset of the answers to the user may be
preferable. A natural requirement for this subset is that it should be
as diverse as possible to reflect the variety of the entire
population. To achieve this, the diversity of a subset is measured
using a metric that determines how different two solutions are and a
diversity function that extends this metric from pairs to sets.  In
the past, several studies have shown that finding a diverse subset from an 
explicitly given set is
intractable even for simple metrics (like Hamming distance) and simple
diversity functions (like summing all pairwise distances). This
complexity barrier becomes even more challenging when trying to output
a diverse subset from a set that is only implicitly given (such as 
the query answers for a given query and a database). Until
now, tractable cases have been found only for restricted problems and
particular diversity functions.

To overcome these limitations, we focus in this work on the notion of
ultrametrics, which have been widely studied and used in many
applications. Starting from any ultrametric $d$ and a 
diversity function $\delta$ extending $d$, we provide sufficient
conditions over $\delta$ for having polynomial-time algorithms to
construct diverse answers. 
To the best of our
knowledge, these conditions are satisfied by all the diversity
functions considered in the literature. Moreover, we complement these
results with lower bounds that show specific cases when these
conditions are not satisfied and finding diverse subsets becomes
intractable. We conclude by applying these results to the evaluation
of conjunctive queries, demonstrating efficient algorithms for finding
a diverse subset of solutions for acyclic conjunctive queries when the
attribute order is used to measure diversity.
\end{abstract}

	\section{Introduction}\label{sec:introduction}

The set of answers to a query may be very large, potentially
overwhelming users when presented with the entire set.
In such cases,
presenting only a small subset of the answers to the user may be
preferable.
Ideally, the selected answers should give the user 
a good overview of the variety present in the complete set of answers.
As was argued in \cite{DBLP:conf/icde/VeeSSBA08}, determining such a small subset 
by sampling will, in general, 
not constitute a satisfactory solution to this problem,
since it would most probably miss interesting but rarely occurring answers. 
Instead, the goal should be to present to the user  a {\em diverse} subset of the answer space
to reflect its variety.

This raises the question of how to define diversity among query results. 
The natural way of defining the diversity of a set of elements (see, e.g., 
\cite{DBLP:conf/aaai/IngmarBST20} for a whole framework on dealing with diversity)
is to first define the 
diversity of 2 elements by a metric (the ``distance'' function) and then to appropriately 
generalize it to arbitrary (finite) sets. 
Both, for the metric and for the generalization to arbitrary sets, many choices exist and, 
as was mentioned in \cite{weitzman1992diversity},
it ultimately depends on the application context which distance and diversity function 
is best suited.

We restrict ourselves here to the relational model. Hence, the answer to a query is a set of tuples and we are interested in outputting 
a subset with a given size $k$ so as to maximize the diversity.
A natural and simple choice for the distance between two tuples is the Hamming distance
(i.e., the number of positions in which the two tuples differ), which was used e.g., 
in the analysis of the diversity of query answers in~\cite{MerklPS23}. A more nuanced point of view was taken 
in~\cite{DBLP:conf/icde/VeeSSBA08,%
DBLP:journals/debu/VeeSA09}, where an ordering of the attributes is assumed and 
tuples are considered as more distant if they differ on an attribute that comes 
earlier in the ordering. This 
idea was exemplified by a car-relation with attributes 
make-model-color-year-description in this order. Hence, for instance, the query engine
would preferably output a subset of cars with different models rather than with different colors. Actually, this distance function is an ultrametric, i.e., a metric $\mathbf{d}$, that 
satisfies the strong triangle inequality 
$\mathbf{d}(a,c) \leq \max\{\mathbf{d}(a,b),\mathbf{d}(b,c)\}$ for any three elements.

For the generalization of the distance to a diversity $\delta$ of sets, 
one can aggregate the pairwise distances, for instance, 
by taking their sum or minimum
(see, e.g., \cite{DBLP:conf/icde/VeeSSBA08,%
DBLP:journals/debu/VeeSA09,%
MerklPS23,%
DBLP:conf/aaai/IngmarBST20}). In this work, we also want to look at a more sophisticated
diversity measure proposed by Weitzman 
in \cite{weitzman1992diversity}, which we will refer to as $\deltaW$. 
It is motivated by the 
goal of capturing the increase of diversity (measured as the minimum
distance from the already chosen elements) when yet another element is added. 
Detailed formal definitions of all concepts mentioned here will be given in 
Section~\ref{sec:preliminaries}.

Aiming at a diverse subset of the answers to a query raises several 
computational problems. The most basic problem is to actually evaluate the
diversity function $\delta$ for a given set $S$ of tuples. Clearly, 
this problem is easy to solve, if $\delta$ is defined by taking 
one
of the standard aggregate functions sum or min
over some efficiently computable metric (such as the Hamming distance).
However, if we take the more sophisticated 
diversity measure $\deltaW$ proposed by Weitzman,
this is not clear any more. In fact, only an exponential algorithm was proposed in 
\cite{weitzman1992diversity} for this task and it was left open, if a polynomial-time 
algorithm exists. We will settle this open question by proving NP-hardness.

Our ultimate goal is to select a 
small subset (say, of size $k$ for given $k > 1$) 
of the query answers so as to maximize the diversity. 
When considering data complexity and restricting ourselves to FO-queries, 
query evaluation is tractable and we may 
assume the entire set $S$ of query answers as {\em explicitly} given. 
Now the goal is to find a subset $S' \subseteq S$ of size $k$ such 
that $\delta(S')$ is maximal. In~\cite{MerklPS23}, it was shown that this task
is NP-hard even for the simple setting where $\delta$ is defined as the sum or as the 
minimum over the pairwise Hamming distances of the tuples. 
Taking the Weitzman diversity clearly makes this task yet more complex. 
It is here that ultrametrics come to the rescue. Indeed, we show tractability of the following problem: given a set $S$ of elements and integer $k > 1$, 
find a subset $S'$ of $S$ such that $\delta(S')$ is maximal, where $\delta$ 
is a diversity function extending an ultrametric and $\delta$ satisfies a certain 
monotonicity property we call \textit{weak subset-monotonicity}.
Moreover, we show that even slightly 
relaxing the monotonicity property immediately leads to NP-hardness.

Things get yet more complex if we consider combined complexity. 
Since the set $S$ of query answers can be exponentially big, we cannot afford to compute it upfront. 
In other words, 
$S$ is only given {\em implicitly} by the database $D$ and query $Q$. 
But the goal remains the same: we want to find a subset $S'$ of $S$ with $|S'| = k$
that maximizes the diversity $\delta(S')$. Since query evaluation is intractable 
for conjunctive queries even without worrying about diversity, 
we now restrict the query language to acyclic conjunctive queries. 
We then manage again to prove tractability for the task of finding a subset $S'$ with 
maximal diversity, provided that $\delta$ is \textit{subset-monotone} -- a restriction slightly stronger than weak subset-monotonicity but which is satisfied by 
$\deltaW$, for example. Again we show tightness of this tractability result by proving that
without this stronger notion of monotonicity, the problem is NP-hard.
Finally, we also identify a kind of middle ground in terms of monotonicity of 
$\delta$ that ensures fixed-parameter tractability when considering $k$ 
as parameter.

\smallskip
\noindent
{\bf Structure of the paper.}
In Section \ref{sec:preliminaries}, we introduce basic notions and formally define the 
computational problems studied here. The complexity of evaluating $\deltaW$ is studied in 
Section \ref{sec:hardness}. 
Fundamental (and well-known) properties of ultrametrics
are recalled in Section \ref{sec:ultrametric}. 
In Sections \ref{sec:explicit-rep} and \ref{sec:implicit-rep},
we study the problem of finding a subset $S' \subseteq S$  maximizing $\delta(S')$ 
for the cases where $S$ is given explicitly or implicitly, respectively. 
In particular, in Section~\ref{sec:implicit-rep}, we first define a general framework that formalizes the notion of 
    an implicitly given set $S$ equipped with an ultrametric. 
The general results for the implicit setting are then studied in 
Section~\ref{sec:acq} for the 
concrete case of combined complexity of query answering for acyclic conjunctive queries. 
We discuss related work in Section~\ref{sec:relatedwork}
and we provide a conclusion and an outlook to future work in 
Section~\ref{sec:conclusions}. The full proofs of all 
result presented here are deferred to the appendix.

	\section{Preliminaries}\label{sec:preliminaries}

\newcommand{\D}{\mathbb{D}}
\newcommand{\arity}{\operatorname{arity}}
\newcommand{\relnames}{\mathcal{R}}
\newcommand{\Var}{\mathcal{X}}
\newcommand{\ba}{\bar{a}}
\newcommand{\bb}{\bar{b}}
\newcommand{\bc}{\bar{c}}
\newcommand{\bx}{\bar{x}}

\paragraph{Sets and sequences} 
We denote by $\bbN$ the set of natural numbers, by $\bbQ$ the set of rational numbers, and by $\bbQgeqz$ the set of non-negative rational numbers.
Given a set $A$, we denote by $\finite(A)$ the set of all non-empty finite subsets of $A$. For $k \in \bbN$, we say that $B \in \finite(A)$ is a \emph{$k$-subset} if $|B| = k$.
We usually use $a$, $b$, or $c$ to denote elements, and $\ba$, $\bb$, or $\bc$ to denote sequences of such elements. For $\ba = a_1, \ldots, a_{k}$, we write $\ba[i] := a_i$ to denote the $i$-th element of~$\ba$ and $|\ba| := k$ to denote the length of $\ba$. Further, given a function $f$ we write $f(\bar{a}) := f(a_1), \ldots, f(a_{k})$ to denote the function applied to each element of $\ba$.

\paragraph{Conjunctive queries} Fix a set $\D$ of data values. A relational schema $\sigma$ (or just schema) is a pair $(\relnames, \arity)$ where $\relnames$ is a set of relation names and $\arity: \relnames \rightarrow \bbN$ assigns 
each name to a number. An $R$-tuple of $\sigma$ (or just a tuple) is a syntactic object
$R(a_1, \ldots, a_{k})$ such that $R \in \relnames$, $a_i \in \D$ for every~$i$,  
and $k = \arity(R)$. We will write $R(\ba)$ to denote a tuple with values $\ba$. A \emph{relational database} $D$ over $\sigma$ is a finite set of tuples over $\sigma$.

Fix a schema $\sigma = (\relnames, \arity)$ and a set of variables $\Var$ disjoint from $\D$. A \emph{Conjunctive Query} (CQ) over $\sigma$ is a syntactic structure of the~form:
\[
Q(\bx) \ \leftarrow \ R_1(\bx_1), \ldots, R_{m}(\bx_{m})  \tag{\dag} \label{eq:cq}
\]
such that $Q$ denotes the answer relation and $R_i$ are relation names in $\relnames$, $\bx_i$ is a sequence of variables in $\Var$, 
$|\bx| = \arity(Q)$, and $|\bx_i| = \arity(R_i)$ for every $i \leq m$. Further, $\bx$ is a sequence of variables appearing in $\bx_1, \ldots, \bx_{m}$. 
We will denote a CQ like (\ref{eq:cq}) by $Q$, where $Q(\bx)$ and $R_1(\bx_1), \ldots, R_{m}(\bx_{m})$ are called the \emph{head} and the \emph{body} of $Q$, respectively. 
Furthermore, we call each $R_i(\bx_i)$ an \emph{atom} of $Q$.

Let $Q$ be a CQ like (\ref{eq:cq}), and $D$ be a database over the same schema~$\sigma$. A \emph{homomorphism} from $Q$ to $D$ is a function $h: \Var \rightarrow \D$ such that $R_i(h(\bx_i)) \in D$ for every $i \leq m$. We define the \textit{answers} of $Q$ over $D$ as the set of $Q$-tuples
\[
\sem{Q}(D) \ := \ \{Q(h(\bx)) \mid \text{$h$ is a homomorphism from $Q$ to $D$}\}.
\]

\paragraph{Diversity setting} Let $\U$ be an infinite set.
We see $\U$ as a \emph{universe} of possible solutions and $S \in \setsU$ as a candidate finite set of solutions that cannot be empty. To determine the diversity of $S$, we first determine how different the pairs of elements are in $S$, which is done through a metric. A \emph{metric} $\dd$ over $\U$ is a function $\dd : \U \times \U \rightarrow \bbQgeqz$ such that $\dd(a,b) = 0$ iff $a = b$, $\dd$ is symmetric (i.e., $\dd(a,b) = \dd(b,a)$), and $\dd$ satisfies the triangle inequality (i.e., $\dd(a,c) \leq \dd(a,b) + \dd(b, c)$). We define the distance of an element $a \in \U$ to a set $S \in \finite(\U)$ as $\dd(a, S) := \min_{b\in S} \dd(a,b)$.

Given a metric $\dd$, a \emph{diversity function} $\delta$ of $\dd$ is a function $\delta : \setsU \rightarrow \bbQgeqz$ such that $\delta(S) = 0$ iff $|S| = 1$, and $\delta(\{a,b\}) = \dd(a,b)$. Note that we see $\delta$ as a function that extends $\dd$ from pairs to sets, and that is $0$ when the set has a single element (i.e., no diversity). Moreover, we impose the 
restriction that $\delta$ should be closed under isomorphism. That is, if $f : \U \rightarrow \U$ is
a bijective function such that $\dd(a,b) = \dd(f(a),f(b))$ for every $a,b \in \U$, then $\delta(S) = \delta(f(S))$ for every $S \in \setsU$.

As proposed in \cite{DBLP:conf/aaai/IngmarBST20}, 
one way of defining a diversity function $\delta$ for a given 
metric $\dd$ over $\U$ is to define an aggregator $f$ that combines the pairwise distances. That is, 
we set 
$\delta(S): = f(\dd(a,b)_{a,b \in S})$.
Common aggregators are sum and min, which give rise to the following diversity functions 
of an arbitrary metric $\dd$%
\footnote{Note that, for $\deltasum(S)$, the distance between any two 
distinct elements $a,b$ is contained twice in this sum, 
namely as $\dd(a,b)$ and $\dd(b,a)$. We could avoid this by 
imposing a condition of the form $a < b$ or by dividing the sum by 2.
However, this is irrelevant in the sequel and, for the sake of simplifying
the notation, we have omitted such an addition.}:
\[
\deltasum(S) := \sum_{a,b \in S} \dd(a,b) \ \ \ \ \text{ and }  \ \ \ \  \deltamin(S) := \min_{a,b \in S \,:\, a \neq b} \dd(a,b).
\]
A more elaborate diversity function is the \emph{Weitzman diversity}
function $\deltaW$~\cite{weitzman1992diversity}, which is recursively
defined as follows:
\begin{equation}
\label{equ:deltaWS}
\deltaW(S) \ := \ \max_{a \in S} \big(\, \deltaW(S \setminus \{a\}) + \dd(a, S \setminus \{a\})  \, \big)
\end{equation}
where $\deltaW(\{a\}) := 0$ is the base case.
In \cite{weitzman1992diversity}, it is shown that $\deltaW$ satisfies several favorable properties. 
For instance, in many application contexts, the ``monotonicity of species" is desirable. That is, adding an 
element 
(referred to as ``species'' in~\cite{weitzman1992diversity})
to a collection should increase its diversity. Now the question is, by how much the diversity 
$\delta$
should increase. Analogously to the first derivative, 
it seems plausible to request  that 
\begin{equation}
\label{equ:derivative}
\delta(S) \ = \ \delta(S \setminus \{a\}) + \dd(a, S \setminus \{a\}) 
\end{equation}
should hold for every $a \in S$.
That is, the additional diversity achieved by adding element $a$ corresponds to 
the distance of $a$ to its closest relative in $S \setminus \{a\}$.
However, as is argued in~\cite{weitzman1992diversity}, since this property can, in general, 
not be satisfied for every element $a$, the diversity function $\deltaW$ provides a 
reasonable approximation to Equation~(\ref{equ:derivative})
by taking the maximum over all $a \in S$.

\nop{*********************
Given a universe $\U$, a metric $\dd$ over $\U$ and a diversity function $\delta$ for $\dd$, we are interested in the following computational problem:
\begin{center}
	\framebox{
		\begin{tabular}{rl}
			\textbf{Problem:} & $\mathtt{Diversity}[\delta]$\\
			\textbf{Input:} & A set $S \subseteq \U$ and $k > 1$ \\
			\textbf{Output:} & $\argmax_{S'\subseteq S: |S'| = k} \delta(S')$
		\end{tabular}
	}
\end{center}
*********************}

\paragraph{Diversity problems} 
When confronted with the task of selecting a {\em diverse} set of elements from an (explicitly or implicitly) given 
set, we are mainly concerned with three problems -- each of them depending on a concrete 
{\em diversity} function $\delta$, which in turn is defined over some universe $\U$ of elements. 
The most basic problem consists in 
computing the diversity for a given $S \in \setsU$:
\begin{center}
		\begin{tabular}{rl}
			\hline\\[-2ex]
			\textbf{Problem:} & $\diversityComputation[\delta]$\\
			\textbf{Input:} & A finite set $S \subseteq \U$ \\
			\textbf{Output:} & $\delta(S)$\\\\[-2.2ex]
			\hline
		\end{tabular}
\end{center}

\smallskip
An additional source of complexity is introduced if the task is to find a subset of $S$ 
with a certain diversity.
\begin{center}
		\begin{tabular}{rl}
			\hline\\[-2ex]
			\textbf{Problem:} & $\diversityExplicit[\delta]$\\
			\textbf{Input:} & A finite set $S \subseteq \U$ and $k > 1$ \\
			\textbf{Output:} & $\argmax_{S'\subseteq S\,:\, |S'| = k} \delta(S')$\\\\[-2.2ex]
			\hline
		\end{tabular}
\end{center}
In light of the previous problem, it is convenient to introduce the following notation: 
we call $S'$ with $S' \subseteq S$ 
a ``{\em $k$-diverse subset of $S$}'', if $S' = \argmax_{A \subseteq S \,:\, |A| = k} \delta(A)$.
In other words, the goal of the $\diversityExplicit[\delta]$ problem is to find a 
$k$-diverse subset of $S$.

\smallskip
Things may get yet more complex, if the set $S \subseteq \U$ from which we want to select a subset with 
maximal diversity is only ``implicitly'' given. 
For us, the most important example of such a setting is when $S$ is the set of answer tuples 
to a given query $Q$ (in particular, an acyclic CQ) over database $D$ and we consider combined complexity. Further settings will be introduced in 
Section~\ref{sec:implicit-rep}. All these settings have in common that 
$S$ might be 
exponentially big and one cannot afford to turn the implicit representation into an explicit one upfront. 
\begin{center}
		\begin{tabular}{rl}
			\hline\\[-2ex]
			\textbf{Problem:} & $\diversityImplicit[\delta]$\\
			\textbf{Input:} & An implicit representation of a finite set $S \subseteq \U$ and $k > 1$ \\
			\textbf{Output:} & $\argmax_{S'\subseteq S \,:\, |S'| = k} \delta(S')$\\\\[-2.2ex]
			\hline
		\end{tabular}
\end{center}

\smallskip
By slight abuse of notation, we will formulate intractability results on the 
three functional problems introduced above in the form of ``$\mathsf{NP}$-hardness'' results. 
Strictly speaking, we thus mean the decision variants of the diversity problems,
i.e., deciding if the diversity $\delta(S)$ is above a given threshold $th$ 
or if a set $S' \subseteq S$ with $\delta(S') \geq th$ exists.

\paragraph{Complexity analysis of algorithms}
For the implementation of our algorithms, we assume the computational model of
Random Access Machines (RAM) with uniform cost measure and addition and subtraction
as basic operations~\cite{DBLP:books/aw/AhoHU74}.
Further, in all the scenarios considered in this paper, a metric $\dd$ is defined
over a countably infinite set $\U$, and the value $\dd(a,b)$ is a
non-negative rational number for every $a,b \in \U$. Thus, we assume
that the codomain of every metric $\dd$ is the set $\bbQgeqz$, which in
particular implies that we have a finite representation for each
possible value of a metric that can be stored in a fixed number of RAM registers. Moreover, although $\dd$ is defined over an
infinite set, we will only need its values for a finite set, and we
assume that $\dd(a,b)$ can be computed in constant time for any pair
$a$, $b$ of elements in this set. Alternatively, one could multiply
the complexity of our algorithms by a parameter $p$ that encapsulates
the cost of computing $\dd(a,b)$ or consider the metric as given by a
look-up table at the expense of a quadratic blow-up of the
input. Neither of these alternatives would provide any additional
insights while complicating the notation or blurring the setting. We
have therefore refrained from adopting one of them.

	\section{Computing diversity is hard}\label{sec:hardness}

\newcommand{\deltaPi}{\delta_{\pi}}

The most basic computational problem considered here is $\diversityComputation[\deltaW]$. 
Clearly, for $\deltasum$ and $\deltamin$, this problem is efficiently solvable. %
Here, we study the complexity of computing the 
diversity $\delta(S)$ of a subset $S \subseteq \U$ 
for the more elaborate Weitzman diversity
measure $\deltaW$. In \cite{weitzman1992diversity}, it was shown that $\deltaW$ can be computed efficiently,
if the distance function $\dd$ is an {\em ultrametric}.
No efficient algorithm was provided for arbitrary distance functions, though, 
and it
was left open whether one exists at all. We settle this open question and give a negative answer by proving $\mathsf{NP}$-hardness of (the decision variant of) this problem. 

\nop{************************
First recall from \cite{weitzman1992diversity} that $\deltaW(S)$ can be computed by 
choosing a permutation $\pi$ of the elements in $S$, constructing $S$ by adding the elements of $S$ in the order of $\pi$, and summing up the distances of each newly added element from the
already existing ones. Then $\deltaW(S)$ is the maximum value attainable over all permutations of $S$.
Formally, for $S = \{ u_1, \dots, u_n \}$, we have: 
\begin{align*}
	\deltaPi(S)  & := \ \dd(u_{\pi(n)}, \{u_{\pi(1)}, \dots, u_{\pi(n-1)} \}) + \mbox{}\\
	&  \mbox{} \quad \ \  \dd(u_{\pi(n-1)}, \{u_{\pi(1)}, \dots, u_{\pi(n-2)}  \}) + \mbox{}\\
	&  \mbox{} \quad \ \  \vdots \\
	& \mbox{} \quad \ \   \ \dd(u_{\pi(3)}, \{u_{\pi(1)}, u_{\pi(2)}  \})+ \mbox{}\\
	& \mbox{} \quad \ \   \ \dd(u_{\pi(2)}, \{u_{\pi(1)} \})\\       
	\deltaW(S)   & := \ \max \big \{ \deltaPi(S) \, \mid \, \pi \mbox{ is a permutation of $S$} \}\big)
\end{align*}
************************}

\begin{theorem}
	\label{thm:weitzman:NP}
	The $\diversityComputation[\deltaW]$ problem of the Weitzman diversity function 
	$\deltaW$ is $\mathsf{NP}$-hard.
\end{theorem}

\begin{proof}[Proof Sketch]
$\mathsf{NP}$-hardness of (the decision variant of) the $\diversityComputation[\deltaW]$ problem 
is shown by reduction from the  \textsc{Independent Set} problem.
Let an arbitrary instance of \textsc{Independent Set} be given by a graph 
$G=(V, E)$ and integer $k$. Let  $|V| = n$. 
Then we set $S = V$ and 
$th = n - k + 2(k-1)$,
and we define the distance function $\dd$ on $S$ as follows: 
	\[
	\dd(u,v) = \begin{cases} 
		0 & \text{if $u = v$} \\  
		1 & \text{if $u$ and $v$ are adjacent in $G$} \\  
		2 & \text{otherwise}  
	\end{cases} 
	\]
It is straightforward to verify that $\dd$ is a metric and that 
$G$ has an independent set of size $k$, if and only if 
$\deltaW(S) \geq th$.  
 \end{proof}

In this work, we are mainly interested in the diversity of sets of tuples -- either from the database itself or sets of tuples resulting from evaluating a query over the database. The above $\mathsf{NP}$-hardness 
proof can be adapted so as to get $\mathsf{NP}$-hardness also for 
the (decision variant of the) 
$\diversityComputation[\deltaW]$ problem if $S$ is a set of tuples, even in a very restricted setting:

\begin{theorem}
	\label{thm:weitzman:NPtuples}
	The $\diversityComputation[\deltaW]$ problem of the Weitzman diversity function 
	$\deltaW$ is $\mathsf{NP}$-hard,
 even if $S$ is a set of tuples of arity 5 and we take the Hamming distance as distance 
 between any two tuples.
\end{theorem}

\begin{proof}[Proof Sketch]
$\mathsf{NP}$-hardness is again shown by reduction from 
the \textsc{Independent Set} problem. As was shown in~\cite{DBLP:conf/ciac/AlimontiK97}, 
the \textsc{Independent Set} problem remains $\mathsf{NP}$-complete even if we restrict the 
graphs to degree 3. Then the crux of the problem reduction is to construct, from a given graph with 
$n$ vertices $\{v_1, \dots, v_n\}$, a set of $n$ tuples
$S = \{t_1, \dots, t_n\}$, such that, for the Hamming distance
$d$
between two tuples $t_i \neq t_j$, we have 
$\dd(t_i,t_j) = 4$ if $v_i,v_j$ are adjacent in $G$ and 
$\dd(t_i,t_j) = 5$ otherwise. For this step, we adapt
a construction that was used in 
\cite{MerklPS23}.
As threshold $th$, we set $th = 4(n-k) + 5(k-1)$. Analogously to 
Theorem~\ref{thm:weitzman:NP}, it can then be shown that $G$ has an independent set of size $k$, if and only if 
$\deltaW(S) \geq th$.  
\end{proof}

When moving from $\diversityComputation[\delta]$ to the 
$\diversityExplicit[\delta]$ problem, of course, the complexity is at least as high. 
So for the Weitzman diversity function $\deltaW$, the intractability clearly carries over. 
However, in case of the $\diversityExplicit[\delta]$ problem, even simpler diversity settings
lead to intractability. More specifically, it was shown 
in~\cite{MerklPS23,DBLP:journals/corr/abs-2301-08848}
that the $\diversityExplicit[\delta]$ problem%
\footnote{Strictly speaking, the problem considered there was formulated as the task 
of finding a set of $k$ answer tuples to an acyclic CQ $Q$ over a database $D$ with diversity $\geq th$. 
$\mathsf{NP}$-hardness was shown for data complexity,  which means that we may assume that the set $S$ of answer tuples is explicitly given since, with polynomial-time effort, 
one can compute $S$.}
is $\mathsf{NP}$-hard even in the simple setting where 
$S$ is a set of tuples of arity 5, considering the Hamming distance and one of the simple
diversity functions $\deltasum$ or $\deltamin$. Therefore, in~\cite{MerklPS23}, 
the parameterized complexity of this problem was considered (with $k$ as parameter) 
and the 
$\diversityExplicit[\delta]$ problem was shown to be fixed-parameter tractable, when 
$S$ is a set of tuples, considering the Hamming distance and very general diversity functions 
$\delta$ satisfying a certain monotonicity property. 

Clearly, for the $\diversityImplicit[\delta]$ problem, things get yet more complex. Indeed, unless $\mathsf{FPT} = \mathsf{W}[1]$, 
 fixed-parameter tractability was ruled out in~\cite{MerklPS23} by showing $\mathsf{W}[1]$-hardness
for the setting where $S$ is the set of answers to an acyclic CQ, considering Hamming distance and one of the 
simple diversity functions $\deltasum$ or $\deltamin$. On the positive side, $\mathsf{XP}$-membership 
was shown for this setting.

In the remainder of this work, we will consider ultrametrics as an important special case of distance
functions. It will turn out that they allow us to prove several positive results for otherwise
hard problems. For instance, the $\diversityImplicit[\delta]$ problem becomes tractable in this case
even when we consider the Weitzman diversity measure $\deltaW$.
 	
	\section{Ultrametrics to the rescue}\label{sec:ultrametric}

In this section, we recall the definition of an ultrametric and present some of its structural properties. The results presented here are well-known in the literature of ultrametric spaces. 
Nevertheless, they are crucial to understand the algorithms for diversity measures shown in the following~sections. 
To make the paper self-contained, we provide proofs of these properties in the appendix.

\paragraph{Ultrametrics} Let $\U$ be a possibly infinite set. An \emph{ultrametric} $\ud$ over $\U$ is a metric over $\U$ that additionally satisfies the \emph{strong triangle inequality}:
\[
\ud(a,c) \leq \max\{\ud(a,b),\ud(b,c)\}.
\]
We use $\dd$ to denote a metric and $\ud$ to denote an ultrametric, thereby making it explicit that we are using an ultrametric.

As an example, consider the following ultrametric for tuples in a relational database with schema~$\sigma$. Let $\U$ be the set of all tuples of $\sigma$. 
Define the metric $\udR$ such that 
$\udR(R(\bar{a}), R(\bar{a})) = 0$, $\udR(R(\bar{a}), S(\bar{b})) = 1$, and   $\udR(R(\bar{a}), R(\bar{a}')) =
2^{-i}$ with $i = \min\{j \mid \bar{a}[j] \neq \bar{a}'[j]\}$, for arbitrary tuples $R(\bar{a})$, $R(\bar{a}')$, and $S(\bar{b})$ of $\sigma$ with $R \neq S, \bar{a} \neq \bar{a}'$. In other words, the distance is $1$ if tuples comes from different relations, and otherwise $2^{-i}$ such that $i$ is the first position where (the arguments of) the tuples differ. One can check that $\udR$
is an ultrametric since, for arbitrary tuples $R(\bar{a})$,
$R(\bar{b})$, $R(\bar{c})$ such that $i$ is the first position where
$R(\bar{a})$ and $R(\bar{c})$ differ, it holds that $R(\bar{a})$ and $R(\bar{b})$ differ at position $i$ or
$R(\bar{b})$ and $R(\bar{c})$ differ at position $i$, so that
\[
\udR(R(\bar{a}),R(\bar{c})) \leq \max\{\udR(R(\bar{a}),R(\bar{b})),\udR(R(\bar{b}),R(\bar{c}))\}. 
\] 
Similarly, the strong triangle inequality holds for tuples of different relations.

\begin{example}
\label{ex:tuplesUltraMetric}
Consider the following running example which is a simplified version taken from~\cite{DBLP:conf/icde/VeeSSBA08} where $\udR$ is used as a metric. In Figure~\ref{fig:cars}, we show the relation $\texttt{CARS}$ that contains car models with the brand (i.e., ``Make''), model, color, and year (in that order). Each row represents a tuple and $t_i$ is the name given to refer to the $i$-th tuple. Then, one can check that $\udR(t_1, t_5) = \sfrac{1}{2}$\, given that $t_1$ is made by Honda and $t_5$ by Toyota. Similarly, $\udR(t_1, t_2) = \sfrac{1}{8}$\, given that $t_1$ and $t_2$ differ in the color, that is, at position $3$ and, thus, $\udR(t_1, t_2) =  2^{-3}$. \qed
\end{example}

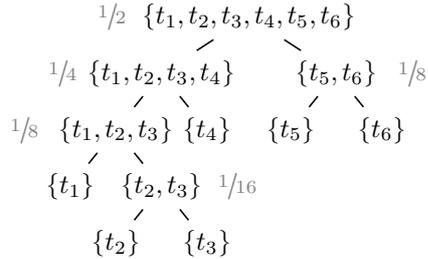
\begin{figure}[t]
	\centering
	\begin{subfigure}[b]{0.40\textwidth}
		\centering
		\texttt{CARS} \vspace{1mm}
		
		\begin{tabular}{|c|llll|} \hline
			& \textbf{Make} & \textbf{Model} & \textbf{Color} & \textbf{Year} \\ \hline
			$t_1$ & Honda & Civic & Green & 2007 \\ \hline
			$t_2$ & Honda & Civic & Black & 2007 \\ \hline
			$t_3$ & Honda & Civic & Black & 2006 \\ \hline
			$t_4$ & Honda & Accord & Blue & 2007 \\ \hline
			$t_5$ & Toyota & Corolla & Black & 2007 \\ \hline
			$t_6$ & Toyota & Corolla & Blue & 2007 \\ \hline
		\end{tabular}
		\vspace{1mm}
		\caption{\label{fig:cars}}
	\end{subfigure}
	\quad \quad \quad \quad
	\begin{subfigure}[b]{0.40\textwidth}
		\centering
		\begin{tikzpicture}[>=stealth, 
			semithick, 
			auto,
			initial text= {},
			initial distance= {4mm},
			accepting distance= {3mm},
			mradius/.style={black!50}]
			
			\path[level distance=0.75cm,
			level 2/.style={sibling distance=1.2cm}]
			node (U) {$\{t_1,t_2,t_3,t_4,t_5,t_6\}$} [sibling distance=2.3cm]
			child { node (n1) {$\{t_1,t_2,t_3,t_4\}$}
				child { node (n3) {$\{t_1,t_2,t_3\}$}
					child { node {$\{t_1\}$}}
					child { node (n4) {$\{t_2,t_3\}$}
						child { node {$\{t_2\}$}}
						child { node {$\{t_3\}$}}
					}
				}
				child { node {$\{t_4\}$}}
			}
			child { node (n2) {$\{t_5,t_6\}$}
				child { node {$\{t_5\}$}}
				child { node {$\{t_6\}$}}
			};
			
			\node [mradius, node distance=1.8cm,left of=U] {$\sfrac{1}{2}$};
			\node [mradius, node distance=1.3cm,left of=n1] {$\sfrac{1}{4}$};
			\node [mradius, node distance=1cm,right of=n2] {$\sfrac{1}{8}$};
			\node [mradius, node distance=1.2cm,left of=n3] {$\sfrac{1}{8}$};
			\node [mradius, node distance=1cm,right of=n4] {$\sfrac{1}{16}$};

		\end{tikzpicture}
		\caption{\label{fig:ultrametrictree}}
	\end{subfigure}	
	\caption{On the left, a relation $\texttt{CARS}$ where each tuple is a car model. On the right, the ultrametric tree of the ultrametric $\udR$ over the tuples $S$ in $\texttt{CARS}$. On one side of each ball $B$ (in grey) we display its radius $\radius(B)$.}
	\label{fig:running-example}
\end{figure}

\paragraph{Structure of ultrametric spaces over finite sets}
Ultrametrics form a class of well-studied metric spaces, which have
useful structural properties. 
For instance, if $\dd$ in 
Equation (\ref{equ:derivative}) from Section \ref{sec:preliminaries}
is an ultrametric, then that equation holds for every 
$a \in S$.
Moreover, every three elements form
an isosceles triangle, namely, for every $a, b, c \in \U$ it holds
that $\ud(a,b) = \ud(a,c)$, or $\ud(a, b) = \ud(b,c)$, or $\ud(a,c)
= \ud(b,c)$. In particular, this implies some well-structured hierarchy
on all balls centered at elements of a finite set, that we introduce
next.
Fix an ultrametric $\ud: \U \times \U \rightarrow \bbQgeqz$  and fix a finite set $S \subseteq \U$. For every $a \in S$ and $r \in \bbQgeqz$, let:
\[
\Ball_S(a, r) := \{b \in S \mid \ud(a, b) \leq r\}.
\]
That is, $\Ball_S(a, r)$ is the (closed) \emph{ball} centered at $a$
with radius $r$. Let $\Ball_S = \{\Ball_S(a,r) \mid a \in S \wedge
r \in \bbQgeqz\}$ be the \emph{set of all balls} of $\ud$ over $S$.
Given that $S$ is finite, $\Ball_S$ is finite as
well. Moreover, the set $\Ball_S$ follows a nested structure given
by the following standard properties of ultrametrics.

\begin{property}\label{prop:ultrametric} 
	(a) For every $B_1, B_2 \in \Ball_S$, it holds that $B_1 \cap B_2 = \emptyset$ or $B_1 \subseteq B_2$ or $B_2 \subseteq B_1$.
	
	\noindent (b) If $\ud(a_1,a_2) \leq r$, then $\Ball_S(a_1,r) = \Ball_S(a_2, r)$.
\end{property}
Property~\ref{prop:ultrametric}(a) implies a nested structure among balls in $\Ball_S$ that we 
can represent as a tree structure as follows. First, for every $B \in \Ball_S$, define the set:
\[
\tparent(B) := \{B' \in \Ball_S \mid B \subsetneq B' \wedge \neg \exists B'' \in \Ball_S \,:\,  B \subsetneq B'' \subsetneq B'\}.
\]
Since $(\Ball_S, \subseteq)$ is a partial order over a finite set, $\tparent(B)$ is non-empty for every $B \neq S$. Moreover, by Property~\ref{prop:ultrametric}(a), we have that $\tparent(B)$ has at most one element. Then, for every $B \in \Ball_S \setminus \{S\}$ we can write $\tparent(B)$ to denote this single element.

We define the \emph{ultrametric tree of $\ud$ over $S$} as the graph $\utree = (\uv, \ue)$ such that $\uv := \Ball_S$ and
\[
\ue \ := \ \{(\tparent(B), B) \mid B \in \Ball_S \setminus \{S\}\}.
\]
Given that $(\Ball_S, \subseteq)$ is a partial order and every $B$ has at most one incoming edge, we have that $\utree$ is a (directed) tree and $S$ is the root of this tree.
Therefore, we can write $\tchildren(B) = \{B' \mid (B,B') \in \ue\}$ to denote the children of $B$ in $\utree$. Note that $B$ is a leaf in the tree $\utree$ iff $B = \{a\}$ for some $a \in S$. Further, if $|B| \geq 2$, then $\tchildren(B)$ forms a partition of $B$ by Property~\ref{prop:ultrametric}(a).

\begin{example}
	Let $S$ be the set of all tuples in the relation $\texttt{CARS}$ of Example~\ref{ex:tuplesUltraMetric}. In Figure~\ref{fig:ultrametrictree}, we display the ultrametric tree of $\udR$ over $S$. One can check in this figure that each leaf contains a single tuple, and the children of each ball form a partition. \qed
\end{example} 

It will be also convenient to relate the distance between elements in $S$ with the radius of the balls in $\Ball_S$. For this purpose, for every $B \in \Ball_S$  we define its \emph{radius} as:
\[
\radius(B) \ := \ \max\{\ud(a,b) \mid a,b \in B\}.
\]
Notice that $\radius(B)$ is well defined since $B$ is a finite set. 
By Property~\ref{prop:ultrametric}(b), we have that $B = \Ball_S(a, \radius(B))$ for every ball $B \in \Ball_S$ and point $a \in B$.
Hence, in what follows, we can use any $a \in B$ as the center of the ball $B$. 

Another crucial property of the radius of $B$ is that it determines the distance between elements of different children of $B$ in $\utree$ as follows. 
\begin{property}\label{prop:radius}
	Let $B_1, B_2 \in \tchildren(B)$ with $B_1 \neq B_2$. Then $\ud(a_1, a_2) = \radius(B)$ for every $a_1 \in B_1, a_2 \in B_2$.
\end{property}  
\begin{example}
	In the ultrametric tree of Figure~\ref{fig:ultrametrictree} we display the radius of each ball. One can check that the root node $S=\{t_1, \ldots, t_6\}$ satisfies that $\radius(S) = \sfrac{1}{2}$ and $\udR(t,t') = \sfrac{1}{2}$ \, for every $t \in \{t_1, \ldots, t_4\}$ and $t' \in \{t_5,t_6\}$. Hence, Property~\ref{prop:radius} holds in this case. \qed
\end{example}

By Properties~\ref{prop:ultrametric} and~\ref{prop:radius}, the ultrametric tree $\utree$ and the radius function $\radius$ completely determine the ultrametric over a finite set $S$ and they will be the starting point for our algorithms. By the following property, we can always construct both in quadratic time over the size of $|S|$. 
\begin{property}\label{prop:build-tree}
	Given an ultrametric $\ud$ and a finite set $S$, we can construct $\utree$ and $\radius$ in $O(|S|^2)$.
\end{property}

Property~\ref{prop:build-tree} above 
can be easily seen by making use of the classical result that 
a minimum-weight spanning tree $\mathcal{T}$ of a graph $G=(S,S\times S)$, where $\ud$ expresses the edge weight, 
can be computed in time $O(|S|^2)$ and the fact that $\utree$ and $\radius$ can be 
easily computed from $\mathcal{T}$. %
On the other hand, it is also easy to see that, in general, one cannot compute the ultrametric tree $\utree$ in time $o(|S|^2)$. Indeed, we just have to consider an ultrametric on 
a set  $S=\{a_1,\dots, a_n\}$ where only one pair $(a_i,a_j)$ of elements has distance 1 and all other pairs have distance 2. Then one cannot compute 
$\utree$ and $\radius$ without finding this pair $(a_i,a_j)$.

	\section{Ultrametrics for explicit representation}\label{sec:explicit-rep}

In this section, we present our first algorithmic results for the $\diversityExplicit[\delta]$
problem, i.e., the problem of finding a $k$-diverse subset of a finite set $S$ given a diversity function $\delta$ of an ultrametric $\ud$. Here, we assume that $S$ is represented explicitly, namely, $S$ is given as a finite list $a_1, \ldots, a_n$.

In order to find tractable scenarios for the explicit case, we
introduce the notion of subset-monotonicity for diversity functions.
A diversity function $\delta$ of a metric $\dd$ is said to
be \emph{subset-monotone} if and only if for every $A \subseteq \U$
and every pair $B, B' \subseteq \U$ such that $B = \{b_1, \ldots,
b_\ell\}$, $B' = \{b_1', \ldots, b_\ell'\}$, $A \cap B = A \cap B'
= \emptyset$, $\delta(B) \leq \delta(B')$, and $\dd(a, b_i) \leq \dd(a,
b_i')$ for every $a \in A$ and $i \in \{1, \ldots, \ell\}$, it holds
that $\delta(A \cup B) \leq \delta(A \cup B')$.  Subset-monotonicity
captures the natural intuition that, if we replace $B$ with $B'$ such
that $B'$ is at least as diverse as $B$ and $B'$ is farther from $A$ than $B$,
then $A \cup B'$ is at least as diverse as~$A \cup B$.

We observe that, for ultrametrics, all diversity functions from Section~\ref{sec:preliminaries} are subset-monotone. The diversity functions $\deltasum$ and $\deltamin$ have this property even for arbitrary metrics.

\begin{proposition}
\label{propo:subsetMonotone}
The diversity functions $\deltasum$ and $\deltamin$ are subset-monotone for every metric. 
If $\deltaW$ extends an ultrametric, then it is also subset-monotone. 
If $\deltaW$ extends an arbitrary metric, then it is, in general, not subset-monotone.
\end{proposition}

Interestingly, if $\delta$ is a subset-monotone diversity function
over an ultrametric $\ud$, then we can always find a $k$-diverse
subset of a finite set $S$ efficiently. In fact, it is possible to
prove this result even if we consider a weaker notion of
subset-monotonicity. A diversity function $\delta$ over a metric $\dd$ is said to
be \emph{weakly subset-monotone} if and only if for every $A \subseteq \U$ and
every pair $B, B' \subseteq \U$ such that $B = \{b_1, \ldots,
b_\ell\}$, $B' = \{b_1', \ldots, b_\ell'\}$, $A \cap B = A \cap B'
= \emptyset$, $\delta(B) \leq \delta(B')$ and $\dd(a, b_i) = \dd(a, b_i')$
for every $a \in A$ and $i \in \{1, \ldots, \ell\}$, it holds that
$\delta(A \cup B) \leq \delta(A \cup B')$. 
In other words,  if we replace $B$ with $B'$ such that $B'$ is
at least as diverse as $B$ and both have the {\em same pairwise distance} 
to $A$,
then $A \cup B'$ is at least as diverse as $A \cup B$. Note that this
weaker version is almost verbatim from subset monotonicity, but with
the condition $\dd(a, b_i) \leq \dd(a, b_i')$ replaced by $\dd(a, b_i) =
\dd(a, b_i')$. Clearly, subset-monotonicity implies weak
subset-monotonicity but not vice versa.

\begin{theorem} \label{theo:explicit-rep}
	Let $\delta$ be a weakly subset-monotone diversity function of an ultrametric $\ud$. Then the problem $\diversityExplicit[\delta]$ can be solved in time $O(k^2 \cdot \fdelta(k) \cdot |S| + |S|^2)$ where $O(\fdelta(k))$ is the time required to compute $\delta$ over a set of size $k$. 
\end{theorem}
Note that $\fdelta(k)\leq k^2$ for $\deltasum$, $\deltamin$, and $\deltaW$, so
we conclude from Theorem~\ref{theo:explicit-rep} that the problem
$\diversityExplicit[\delta]$ can be solved in polynomial time for
these fundamental diversity functions.
\begin{proof}[Proof Sketch of Theorem~\ref{theo:explicit-rep}]
The main ideas of the algorithm for Theorem~\ref{theo:explicit-rep} are the following.
Let $S$ be a finite subset of the universe $\U$, $\ud$ an ultrametric over $\U$, and $\delta$ a weakly subset-monotone diversity function of $\ud$. By Proposition~\ref{prop:build-tree}, we can construct an ultrametric tree $\utree$ of $S$ in time $O(|S|^2)$. For the sake of simplification, assume that $\utree$ is a binary tree. Otherwise, one can easily extend the following ideas to the non-binary case. %
For each vertex $B$ (i.e., a ball) of $\utree$, we maintain a function $\cand_B\colon \{0,\dots,\min\{k, |B|\}\} \rightarrow 2^{B}$ where $\cand_B(i)$ is an 
$i$-diverse subset of $B$. That is, for every $i \in \{0,\dots,\min\{k, |B|\}\}$ we have 
$
\cand_{B}(i) \ := \  \argmax_{A \subseteq B: |A| = i} \delta(A).
$
Clearly, if we can compute $\cand_{S}$ for the root $S$ of $\utree$, then $\cand_{S}(k)$ is a $k$-diverse subset for $S$. 

The algorithm follows a dynamic programming approach, computing $\cand_B$ for each $B \in \Ball_S$ in a bottom-up fashion over $\utree$. For every ball $B$, we can easily check that $\cand_B(0) = \emptyset$ and $\cand_B(1) = \{a\}$ for some $a \in B$. In particular, $\cand_{\{a\}} = \big\{0 \mapsto \emptyset, 1 \mapsto \{a\} \big\}$ is our base case for every leaf $\{a\}$ of $\utree$.
For an inner vertex $B$ of $\utree$, the process is a bit more involved. Let $B_1$ and $B_2$ be the two children of $B$ in $\utree$ and assume that we have already computed $\cand_{B_1}$ and $\cand_{B_2}$. 
We claim that, for every $i \in \{0,\dots,\min\{k, |B|\}\}$, we can calculate $\cand_{B}(i)$ as $\cand_{B}(i) = \cand_{B_1}(i_1) \cup \cand_{B_2}(i_2)$, where $i_1,i_2$ are obtained as follows:
\[
(i_1,i_2) = \argmax_{(j_1,j_2): j_1+j_2 = i} \delta(\cand_{B_1}(j_1) \cup \cand_{B_2}(j_2)).
\]

Intuitively, when maximizing diversity with $i$ elements of $B$, one must try all combinations 
of a $j_1$-diverse subset from $B_1$ and a $j_2$-diverse subset from $B_2$, 
such that $i = j_1 + j_2$ holds.
That is, the best elements to pick from $B_1$ and $B_2$ are found in $\cand_{B_1}$ and $\cand_{B_2}$, respectively.

To see that the claim holds, let $A$ be a subset of $B$ with $i$-elements maximizing $\delta(A)$. Define $A_1 := A \cap B_1$ and $A_2 := A \cap B_2$, and their sizes $i_1 := |A_1|$ and $i_2 := |A_2|$, respectively. Due to the optimality of $\cand_{B_1}$, we know that $\delta(A_1) \leq \delta(\cand_{B_1}(i_1))$. Further, $\ud(a_1, a_2) = \ud(a_1',a_2) = \radius(B)$ for every $a_1 \in A_1, a_1' \in \cand_{B_1}(i_1), a_2 \in A_2$ by Property~\ref{prop:radius}. Then the conditions of weak subset-monotonicity are satisfied and $\delta(A_1 \cup A_2) \leq \delta(\cand_{B_1}(i_1) \cup A_2)$. Following the same argument, we can conclude that $\delta(\cand_{B_1}(i_1) \cup A_2) \leq \delta(\cand_{B_1}(i_1) \cup \cand_{B_2}(i_2))$, proving that $\cand_{B_1}(i_1) \cup \cand_{B_2}(i_2)$ is optimal. 

By the previous ideas, the desired algorithm 
with time complexity  $O(k^2 \cdot \fdelta(k) \cdot |S| + |S|^2)$
for solving the $\diversityExplicit[\delta]$ problem follows.
\end{proof}
By Theorem~\ref{theo:explicit-rep}, we get the following result for finding $k$-diverse outputs of CQ query evaluation, where
$\fQ(D) \leq |D|^{|Q|}$.
\begin{corollary}\label{cor:CQ-exp-rep}
Let $\ud$ be an ultrametric over tuples, $\delta$ be a weakly subset-monotone diverse function of $\ud$, and $Q$ be a fixed CQ (i.e., data complexity). Given a relational database $D$, and a value $k$ (in unary), we can compute a $k$-diverse subset of $\sem{Q}(D)$ with respect to $\delta$ in time $O(k^2 \cdot \fdelta(k) \cdot |\sem{Q}(D)| + |\sem{Q}(D)|^2 + \fQ(D))$ where $\fQ(D)$ is the time required to evaluate $Q$ over $D$.
\end{corollary}
An open question is whether we can extend
Theorem~\ref{theo:explicit-rep} beyond weakly subset-monotone diversity
functions (extending ultrametrics). We provide here a partial answer by
focusing on monotone diversity functions. A diversity function
$\delta$ of a metric $\dd$ is said to be \emph{monotone} if and only if for every $A,
A' \subseteq \U$ such that $A = \{a_1, \ldots, a_\ell\}$, $A'
= \{a_1', \ldots, a_\ell'\}$, and $\dd(a_i, a_j) \leq \dd(a_i', a_j')$ for
every $i, j \in \{1, \ldots, \ell\}$, it holds that
$\delta(A) \leq \delta(A')$. Monotone diversity functions were
considered in~\cite{MerklPS23} as a general class of natural diversity
functions. 
In the next result, we show that monotonicity of the 
diversity function $\delta$ over some ultrametric is, in general, 
not enough to make the $\diversityExplicit[\delta]$ problem tractable.

\begin{theorem}\label{theo:exp-hardness}
	The $\diversityExplicit[\delta]$ problem is $\mathsf{NP}$-hard even for a monotone, efficiently computable diversity function $\delta$ based on an~ultrametric.
\end{theorem}

That is, the previous result implies that there are monotone diversity functions beyond the weakly subset-monotone class where the algorithmic strategy of Theorem~\ref{theo:explicit-rep} cannot be used.

	\section{Ultrametrics for implicit representations}\label{sec:implicit-rep}

We move now to study the case when $S$ is represented implicitly. Our motivation for implicit representations is to model the query evaluation setting: we receive as input a query $Q$ and a database $D$, and we want to compute a $k$-diverse subset of $S = \sem{Q}(D)$. The main challenge is that $S$ could be of exponential size concerning $|Q|$ and $|D|$; namely, $S$ is implicitly encoded by $Q$ and $D$, and it is not efficient first to compute $S$ to find a $k$-diverse subset of $S$. To formalize this setting in general, given a universe $\U$, we say that an \emph{implicit schema over $\U$} is a tuple $(\irI, \sem{\cdot})$ where $\irI$ is a set of objects called implicit representations, and $\sem{\cdot}$ is a function that maps every implicit representation $I \in \irI$ to a finite subset of $\U$. Further, we assume the existence of a size function $|\cdot|: \irI \rightarrow \bbN$ that represents the size $|I|$ of each implicit representation $I \in \irI$. For example, $\irI$ can be all pairs $(Q, D)$ where $Q$ is a CQ and $D$ is a relational database, $\sem{\cdot}$ maps each pair $(Q,D)$ to $\sem{Q}(D)$, and $|(Q,D)| = |Q| + |D|$. Note that we do not impose any restriction on  the number of elements of $\sem{I}$, so it can be arbitrarily large with respect to $|I|$. Our goal in this section is to compute efficiently, given an implicit representation $I \in \irI$ and $k > 1$, a $k$-diverse subset of $\sem{I}$ with respect to a diversity function $\delta$ of an ultrametric $\ud$.

Given this general scenario, we need a way to navigate through the
elements of $\sem{I}$. In particular, we need a way to navigate the
ultrametric tree $\utree[\sem{I}]$ of $\ud$ over $\sem{I}$.
Like $\sem{I}$, $\utree[\sem{I}]$ could be arbitrarily
large with respect to $|I|$, so it could be unfeasible to construct
$\utree[\sem{I}]$ explicitly. For this reason, we will assume that our
implicit schemas admit some efficient algorithms for
traversing ultrametric trees. Formally,
an \emph{implicit ultrametric tree} for an implicit schema
$(\irI, \sem{\cdot})$ consists of three algorithms
$(\irroot, \irchildren, \irmember)$ such that, given an implicit representation
$I \in \irI$ and a ball $B \in \Ball_{\sem{I}}$:
\begin{enumerate}
	\item $\irroot(I)$ computes the root of $\utree[\sem{I}]$ in polynomial time with respect to $|I|$;
	\item $\irchildren(I, B)$ enumerates all  children of $B$ in $\utree[\sem{I}]$ with polynomial delay w.r.t.~$|I|$; and 
	\item $\irmember(I, B)$ outputs one solution in $B$ in polynomial time with respect to $|I|$.
\end{enumerate}
Further, when we say that a method receives a ball $B$ as input or enumerates $B$ as output, it means an ID representing the ball $B$.
Recall that $B \subseteq \sem{I}$ and then $B$ could be large with respect to $|I|$. For this reason, methods $\irroot(I)$ and $\irchildren(I, B)$ output IDs representing balls in $\Ball_{\sem{I}}$, that one later uses to call $\irchildren$ and $\irmember$. Here, we assume that each ID has the size of one register of the RAM or a small number of registers that one can bound by some parameter on $I$ (e.g., the arity of query answers if $(\irI, \sem{\cdot})$ models the query evaluation setting). 
For example, in the next section we show that such a representation exists in the case of acyclic CQs.
Regarding performance, we say that we can compute an implicit ultrametric tree in time $O(\fT(I))$, for some function $\fT$, if the running time of $\irroot$ and $\irmember$, and the delay of $\irchildren$ are in $O(\fT(I))$. 
Note that we can always assume that $\fT(I)\leq |I|^\ell$ for some constant $\ell$.

Unlike the results presented in Section \ref{sec:explicit-rep}, there
is a difference in the complexity of the problem $\diversityImplicit[\delta]$
depending on whether a diversity function is subset-monotone or weakly
subset-monotone. First, it is possible to show that
$\diversityImplicit[\delta]$ is tractable when restricted to the class of subset-monotone diversity functions.

\begin{theorem}\label{theo:ptime-implicit-rep}
Let $(\irI, \sem{\cdot})$ be an implicit schema and $\ud$ an ultrametric over a common universe $\U$ that admit an implicit
ultrametric tree, and $\delta$ be a subset-monotone diversity
function of $\ud$.
Further, assume that the running time of computing $\delta$ over a $k$-subset of $\U$ is bounded by $O(\fdelta(k))$, and we can compute the implicit ultrametric tree in time $O(\fT(I))$. Then, the problem $\diversityImplicit[\delta]$ can
be solved in time $O(k \cdot \fT(I) + k^2 \cdot \fdelta(k))$. %
\end{theorem}

\begin{proof}[Proof sketch]
The algorithm of Theorem~\ref{theo:ptime-implicit-rep} %
follows a different approach than
the algorithm for explicit representation
in~Theorem~\ref{theo:explicit-rep}. Instead of processing the
ultrametric tree in a bottom-up fashion, it proceeds top-down by using
the interface of the implicit representation of the ultrametric tree
and following a greedy approach: 
we maintain a set of solutions $S$ and a set of candidates $L$.
At each step we add a $l\in L$ to $S$ that maximizes the ``incremental'' diversity, i.e., crucially
\[
\delta(S\cup \{l\}) = \max_{l'\in L}\delta(S\cup \{l'\}) = \max_{l''\in \sem{I}}\delta(S\cup \{l''\}).
\]
Due to subset-monotonicity, proceeding greedily is correct and due to the implicit ultrametric tree, we can maintain $L$ efficiently, i.e., it requires time $O(\fT(I))$ at each step.
\end{proof}

Unfortunately, and in contrast with Theorem~\ref{theo:explicit-rep} in Section~\ref{sec:explicit-rep}, the tractability of the problem $\diversityImplicit[\delta]$ no longer holds if we consider weakly subset-monotone diversity functions.
To prove this, let $\dd$ be any metric over a universe $\U$, and define the \emph{sum-min diversity function} $\deltasummin$ as follows. For every set $S \in \finite(\U)$: $\delta(S) = 0$ if $|S| = 1$, and $\delta(S)$ is given by the following expression if $|S| > 1$:
\[
\deltasummin(S) \ := \ \sum_{a \in S} \dd(a, S \setminus \{a\}) \ = \ \sum_{a \in S} \, \min_{b \in S \,:\, b \neq a} \dd(a, b).
\]
Intuitively, $\deltasummin$ is summing the contribution of each element $a \in S$ to the diversity of $S$, namely, how far is $a$ from the other elements in $S$. One can see $\deltasummin$ as a non-recursive version of the Weitzman diversity function. Like the other diversity functions used before, we can prove that $\deltasummin$ is also a weakly subset-monotone diversity function. 
\begin{proposition}\label{prop:sum-min}
	$\deltasummin$ is weakly subset-monotone if it extends an ultrametric. 
\end{proposition}

As we show next, $\deltasummin$ serves as an example of a
weakly subset-monotone diversity function of an ultrametric~$\ud$ for which
one can find an implicit representation that admits an implicit
ultrametric tree, but where it is hard to find a $k$-diverse subset.
\begin{theorem}\label{theo:hardness-sum-min}
	There exists an implicit schema $(\irI, \sem{\cdot})$ and an ultrametric $\ud$ over a common universe $\U$ which admit an implicit ultrametric tree but for which $\diversityImplicit[\deltasummin]$ is $\mathsf{NP}$-hard.  
\end{theorem}

A natural question is whether subset-monotonicity can be relaxed in
other ways to allow for the tractability of the problem
$\diversityImplicit[\delta]$. We conclude this section by providing an
answer to this question when considering the notion of
fixed-parameter tractable (FPT) algorithms. 
More precisely, we say that a diversity function $\delta$ of a metric
$\dd$ over $\U$ is \emph{weakly monotone} if, and only if, for every set
$A \subseteq \U$ and every $b, b' \in \U$ such that $A \cap \{b,b'\}
= \emptyset$ and $\dd(a, b) \leq \dd(a, b')$ for every $a \in A$, it holds that
$\delta(A \cup \{b\}) \leq \delta(A \cup \{b'\})$. In other words, if
we extend $A$ with $b$ or $b'$, then $A \cup \{b'\}$ will be at least as diverse as
$A \cup \{b\}$ given that $b'$ is farther from $A$ than $b$.

Every subset-monotone function is also weakly monotone, but weak
subset-monotonicity does not necessarily imply weak
monotonicity. However, one can check that all diversity functions
used in this paper are weakly monotone, in particular, $\deltasummin$.

\begin{theorem}\label{theo:fpt-implicit-rep}
For every implicit schema $(\irI, \sem{\cdot})$ and
ultrametric $\ud$ over a common universe\, $\U$ which admits an implicit ultrametric
tree, and for every computable weakly monotone diversity function $\delta$ of
$\ud$, the problem $\diversityImplicit[\delta]$ is fixed-parameter
tractable in $k$.
\end{theorem}

\begin{proof}[Proof sketch]
Our FPT algorithm %
navigates the ultrametric tree top-down by using its implicit representation, but this time it considers balls up to distance $k$ from the root.
In the worst case, the ultrametric tree is binary and we have to consider all balls up to depth $k$.
We then iterate over these balls (up to $2^k$ many), obtaining a set $S$ by selecting one solution form each ball. We show that if $\delta$ is weakly monotone, a $k$-diverse subset of $S$ is a $k$-diverse subset of $\sem{I}$. 
\end{proof}

	\section{Efficient computation of diverse answers to ACQs}\label{sec:acq}

In this section, we use the results for implicit representations from the previous section to obtain efficient algorithms for finding $k$-diverse subsets of the answers to acyclic CQ (ACQ) with respect to the ultrametric $\udR$ over tuples presented in Section~\ref{sec:ultrametric}. 
Note that here we study algorithms for ACQ in combined complexity (i.e., the query $Q$ is not fixed), in contrast to Corollary~\ref{cor:CQ-exp-rep} whose analysis is in data complexity (i.e., $Q$ is fixed). 
In the following, we start by recalling the definition of ACQ and discussing the ultrametric $\udR$. Then, we show our main results concerning computing diverse query answers for ACQ. 

Acyclic CQ is the prototypical subclass of conjunctive queries that allow for tractable query evaluation 
(combined complexity)~\cite{yannakakis1981algorithms,bagan2007acyclic}. We therefore also take ACQ as the 
natural starting point in our effort to develop efficient algorithms for finding $k$-diverse sets. Let $Q$ be a CQ like in (\ref{eq:cq}). A \emph{join tree} for $Q$ is a labeled tree $T = (V, E, \lambda)$ where $(V, E)$ is a undirected tree and $\lambda$ is a bijective function from $V$ to the atoms $\{R_1(\bx_1), \ldots, R_{m}(\bx_{m})\}$. Further, a join tree $T$ must satisfy that each variable $x \in \Var$ forms a connected component in $T$, namely, the set $\{v \in V \mid \text{$x$ appears in the atom $\lambda(v)$}\}$ is connected in~$T$. Then $Q$ is \emph{acyclic} iff there exists a join tree for $Q$. 
Also, we say that $Q$ is a \emph{free-connex ACQ} iff $Q$ is acyclic and the body together with the head of $Q$ is acyclic (i.e., admits a join tree).

For our algorithmic results over ACQ, we restrict to the ultrametric $\udR$ over tuples of a schema~$\sigma$, previously defined in Section~\ref{sec:ultrametric}. 
Arguably, $\udR$ is a natural ultrametric for comparing tuples of relations that has been used for computing diverse subsets in previous work~\cite{DBLP:journals/debu/VeeSA09,DBLP:conf/icde/VeeSSBA08}.
Let $Q$ be a CQ with head $Q(\bar{x})$. 
Note that the variable order $\bar{x}$ in $Q(\bar{x})$ is important for measuring the diversity of subsets of $\sem{Q}(D)$ for some database $D$. Concretely, if we have two CQ $Q$ and $Q'$ with the same body but with different orders in their heads, then diversity of the subsets of $\sem{Q}(D)$ and $\sem{Q'}(D)$ could be totally different.

By using Theorem~\ref{theo:ptime-implicit-rep} over ACQ and the ultrametric $\udR$, we can find quasilinear time algorithms with respect to $|D|$ for finding $k$-diverse sets of query answers.

\begin{theorem}\label{cor:acq-evaluation}
	Let $\delta$ be a subset-monotone diversity function of the ultrametric $\udR$ such that the running time of computing $\delta$ over a set of size $k$ is bounded by $O(\fdelta(k))$. Given an ACQ $Q$, a relational database $D$, and a value $k$ (in unary), a $k$-diverse subset of $\sem{Q}(D)$ with respect to $\delta$ can be computed in time $O(k \cdot |Q|  \cdot |D| \cdot \log(|D|) + k^2 \cdot \fdelta(k))$.
\end{theorem}
\begin{proof}[Proof Sketch]
To apply Theorem~\ref{theo:ptime-implicit-rep}, we need to describe an implicit ultrametric tree for tuples in $\sem{Q}(D)$. A helpful property of ultrametric $\udR$ is that balls in $\Ball_{\sem{Q}(D)}$ can be represented by the common prefix of their tuples. More specifically, for every ball $B \in \Ball_{\sem{Q}(D)}$ there exist values  $c_1, \ldots, c_i$ such that $B = \{Q(\bar{a}) \in \sem{Q}(D) \mid \forall j \leq i. \,\bar{a}[j] = c_j\}$.  
Then, we can traverse the ultrametric tree of $\udR$ over $\sem{Q}(D)$, by managing partial outputs (i.e., prefixes) of $\sem{Q}(D)$. 

Let $Q$ be an ACQ like~(\ref{eq:cq}). To arrive at the methods $\irroot$, $\irchildren$, and $\irmember$, we modify Yannakakis algorithm~\cite{yannakakis1981algorithms} to find the following data values that extend a given prefix. 
Concretely, given values $c_1, \ldots, c_i$ that represent a ball $B$, we want to find all values $c$ such that $c_1, \ldots, c_i, c$ is the prefix of a tuple in $\sem{Q}(D)$. For this, we can consider the subquery $Q'(\bar{x}[i+1])) \leftarrow R_1(h(\bx_1)), \ldots, R_{m}(h(\bx_{m}))$ where $h$ is a partial assignment that maps $h(\bar{x}[j]) = c_j$ for every $j \leq i$ and $h(x) = x$ for any other variable $x \in \Var$.  The subquery $Q'$ is also acyclic and returns all the desired values $c$, such that $c_1 \ldots, c_i, c$ represents a child of $B$. Thus, running Yannakakis algorithm over $Q'$ and $D$, we can compute $\irchildren$ in time $O(|Q|\cdot |D| \cdot \log(|D|))$ and similarly for the methods $\irroot$ or $\irmember$.
\end{proof}

A natural next step to improve the running time of Theorem~\ref{cor:acq-evaluation} is to break the dependency between $k$ and $|Q|\cdot |D|\cdot \log(|D|)$. Towards this goal, we take inspiration from the work of Carmeli et al.~\cite{DBLP:journals/tods/CarmeliTGKR23}, which studied direct access to ranked answers of conjunctive queries. In this work, the algorithmic results also depend on the attribute order, characterizing which CQs and orders admit direct access to the results. 
For this characterization, the presence of a disruptive trio in the query is crucial. Let $Q$ be a CQ like (\ref{eq:cq}). We say that two variables $x$ and $y$ in $Q$ are neighbors if they appear together in $Q$ in some atom. Then we say that three positions $i, j, k$ (i.e., variables $\bx[i]$, $\bx[j]$, and $\bx[k]$) in the head of $Q$ form a \emph{disruptive trio} iff $\bar{x}[i]$ and $\bar{x}[j]$ are not neighbors in $Q$, and $\bar{x}[k]$ is a neighbor of $\bar{x}[i]$ and $\bar{x}[j]$ in $Q$, but $i <k$ and $j < k$ (i.e., $\bar{x}[k]$ appears after $\bar{x}[i]$ and $\bar{x}[j]$). For example, for the query $Q(x_1,x_2,x_3,x_4) \leftarrow R(x_1,x_2), S(x_2, x_4), T(x_4, x_3)$, the positions $2, 3, 4$ form a disruptive trio, but $1,2,3$ do not.   

In the following result, we show that free-connex ACQ and the absence of a disruptive trio are what we need to get better algorithms for computing $k$-diverse subsets.

\begin{theorem}\label{theo:disruptive-trio}
	Let $\delta$ be a subset-monotone diversity function of the ultrametric $\udR$ such that the running time of computing $\delta$ over a set of size $k$ is bounded by $O(\fdelta(k))$. Given a free-connex ACQ $Q$ without a disruptive trio, a relational database $D$, and a value $k$ (in unary), a $k$-diverse subset of $\sem{Q}(D)$ with respect to $\delta$ can be computed in time $O(|Q| \cdot |D|\cdot\log(|D|)  + k \cdot |Q| + k^2 \cdot \fdelta(k))$.
\end{theorem}
\begin{proof}[Proof Sketch]
	Similar to Theorem~\ref{cor:acq-evaluation}, we use the prefix of tuples to represent balls and take advantage of the structure of a join tree and the absence of disruptive trios to implement an index over $D$. By~\cite{DBLP:journals/tods/CarmeliTGKR23}, the absence of disruptive trios ensures the existence of a \emph{layered join tree} whose layers follow the order of the variables in the head of $Q$. Then, by using the layered join tree, we can compute an index in time $O(|Q| \cdot |D|\cdot \log(|D|))$ that we can use to calculate the next data value in a prefix, like in~Theorem~\ref{cor:acq-evaluation}, but now in time $O(|Q|)$.
	Thus, after a common preprocessing phase, $\irroot, \irchildren$ and $\irmember$ run in time $O(|Q|)$.
\end{proof}

Note that in data complexity, the only remaining non-(quasi)linear term in Theorem~\ref{theo:disruptive-trio} is $k^2\cdot \fdelta(k)$ which arises since we have to reevaluate $\delta$ at each step to find the next greedily best pick.
For some specific diversity function this may not be necessary. 
As an example, for the Weitzman diversity function $\deltaW$ we can get rid of this term $k^2\cdot \fdelta(k)$ by smartly keeping track of which answer maximizes the diversity next.

\begin{theorem}\label{theo:disruptive-trio-W} 
	Let $\deltaW$ be the Weitzman diversity function over the ultrametric $\udR$. 
	Given a free-connex ACQ $Q$ without a disruptive trio, a relational database $D$, and a value $k$ (in unary), a $k$-diverse subset of $\sem{Q}(D)$ with respect to $\deltaW$ can be computed in time $O(|Q| \cdot |D|\cdot \log(|D|)  + k \cdot |Q|)$.
\end{theorem}

	\section{Related work}\label{sec:relatedwork}

\paragraph{Diversification}
Aiming for a small, {\em diverse} subset of the solutions has been adopted in many areas 
as a viable strategy of dealing 
with a solution space that might possibly be overwhelmingly big.
This is, in particular, the case in data mining, information retrieval, and web science, where
the term {\em ``diversification of search results''} is commonly used for the process of 
extracting a small diverse subset from a huge set of solutions, see e.g., 
\cite{DBLP:conf/kdd/SuD0W22,%
DBLP:conf/sigir/DemidovaFZN10,%
DBLP:journals/ftir/SantosMO15,%
DBLP:conf/www/GollapudiS09,%
DBLP:conf/www/SantosMO10} and the surveys  
\cite{DBLP:journals/kais/ZhengWQLG17,TKDE2024Survey}.
The diversity of solutions has also been intensively studied by the Artificial Intelligence (AI) community. 
Notably, this is the case in subfields of AI which are most closely related to 
database research, namely constraint satisfaction (recall that, from a logical point of view, 
solving constraint satisfaction problems and evaluating 
conjunctive queries are equivalent tasks) 
\cite{DBLP:conf/aaai/IngmarBST20,%
DBLP:conf/aaai/HebrardHOW05,%
DBLP:conf/ijcai/HebrardOW07,%
DBLP:conf/ijcai/PetitT15}
and answer set programming (which corresponds to datalog with unrestricted negation under stable model semantics) \cite{DBLP:journals/tplp/EiterEEF13}. 
In the database community, the diversification of query answers has been on the agenda for over a 
decade: the computation of diverse query results was studied in ~\cite{DBLP:journals/debu/VeeSA09,%
DBLP:conf/icde/VeeSSBA08} for relational data and in 
\cite{DBLP:journals/pvldb/LiuSC09} for XML data.
In \cite{DBLP:journals/pvldb/VieiraRBHSTT11}, a system was presented 
with an extension of SQL to allow for requesting diverse answers. 
In~\cite{DBLP:journals/tods/DengF14}, query result diversification is studied as 
a ``bi-criteria'' optimization problem that aims at finding $k$ query answers that maximize both, the diversity and the relevance of the answers.
Recent publications witness the renewed interest in the diversity of query answers
by the database systems community~\cite{DBLP:journals/vldb/NikookarEBSAR23,%
DBLP:journals/vldb/IslamAAR23} and the database theory community~\cite{ArenasCJR21,MerklPS23}.

\paragraph{Measuring diversity}
In \cite{DBLP:conf/aaai/IngmarBST20}, a whole framework for dealing with the 
diversity of subsets of the solutions
has been proposed. There, diversity is defined by first defining the distance between two solutions and then 
combining the pairwise distances via an aggregate function such as, for instance, 
sum, min, or max.
In \cite{weitzman1992diversity}, 
a more sophisticated way of aggregating pairwise distances has 
led to the definition of the diversity function $\deltaW$, 
which we have had a closer look at in our work.
Yet more complexity was introduced in 
\cite{DBLP:conf/kr/SchwindOCI16}, where the diversity of a subset of solutions not only takes the 
relationships between the chosen solutions into account but also their relationship with the solutions 
excluded from the subset.

For the distance between two solutions, any metric can be used. As is argued in \cite{weitzman1992diversity},
it ultimately depends on the application context which distance function (and, consequently, which diversity function)
is appropriate. Note that 
diversity and similarity can be seen as two sides of the same medal. Hence, all 
kinds of similarity measures studied in the 
data mining and 
information retrieval communities are, in principle, also candidates for the distance function;
e.g., the Minkowski distance with Manhattan and Euclidean
distance as important special cases as well as Cosine distance when solutions are represented as vectors, the edit
distance for solutions as strings, or the Jaccard Index for solutions as sets, etc., 
see e.g., \cite{book/GanEtAl2007,annDataScienceXuTian15}.

\paragraph{Ultrametrics}
The study of ultrametrics started in various areas of mathematics (such as real analysis, number theory, and 
general topology -- see \cite{ultrametrization}) 
in the early 20th century. Ultrametrics are particularly well suited for hierarchical clustering and, as such, 
they have many applications in various sciences such as 
psychology, physics, and biology
(see, e.g., \cite{UltrametricModelOfMind,%
UltrametricPhysics,%
weitzman1992diversity}) and, of course, 
also in data mining, 
see e.g., \cite{DataMiningHandbookArticle2007}.
Ultrametrics have also been used in database research and 
related areas: 
in \cite{DBLP:journals/tcs/HitzlerS03}, 
several convergence criteria for the fixed-point iteration of datalog programs (or, more generally, logic programs) with negation are defined. To this end, the set of possible ground atoms is divided into {\em levels} and 
two sets of ground atoms are considered as more diverse if they differ on an earlier level. Clearly, this is 
an ultrametric. 
In \cite{DBLP:journals/debu/VeeSA09,%
DBLP:conf/icde/VeeSSBA08},
the distance between tuples is defined by imposing an order on the attributes and considering two 
tuples as more diverse if they differ on an earlier attribute in this ordering
(see also Example~\ref{ex:tuplesUltraMetric} in the current paper). Again, this is clearly an ultrametric, even though it was not 
explicitly named as such in \cite{DBLP:journals/debu/VeeSA09,%
DBLP:conf/icde/VeeSSBA08}.

\paragraph{Computing diversity}
It should be noted that searching for diverse sets 
is, in general, an intractable problem. 
For instance, in \cite{MerklPS23}, NP-completeness of the 
$\diversityExplicit[\delta]$ problem%
\footnote{Strictly speaking, in \cite{MerklPS23}, $S$ was defined as the result set of an FO-query.
However, 
since data complexity was considered, we can of course compute $S$ upfront in polynomial time and may,
therefore, assume $S$ to be explicitly given.}
was proved even for the simple setting
where $S$ is a set of tuples of arity five and defining the 
diversity via the sum or min of the pairwise Hamming distances.  
Consequently, 
approximations or heuristics are typically proposed 
to compute diverse sets (see \cite{TKDE2024Survey} for a very recent survey on 
diversification methods).
In \cite{DBLP:journals/ai/BasteFJMOPR22}, the parameterized version of the $\diversityExplicit[\delta]$ problem was studied for cases where 
the problem of 
deciding the existence of a solution is fixed-parameter tractable (FPT)
w.r.t.\ the treewidth. 
As a 
    prototypical problem, the Vertex Cover problem was studied and it was shown that, when defining the
    diversity as the sum of the pairwise Hamming distances, then the $\diversityExplicit[\delta]$
problem is FPT w.r.t.\ the treewidth $w$ and the size $k$ of the desired diversity set. 
Moreover, it was argued in 
 \cite{DBLP:journals/ai/BasteFJMOPR22}
that analogous FPT-results for the $\diversityExplicit[\delta]$ problem apply to 
virtually any problem where the decision of the existence of a solution 
is FPT w.r.t.\ the treewidth.
In \cite{MerklPS23}, the $\diversityExplicit[\delta]$ problem was shown FPT
w.r.t.\ the size $k$ of the desired diversity set, when considering $S$ as a set of tuples and defining 
the diversity via a monotone aggregate function 
over the pairwise Hamming distances of the tuples.

However, {\em tractable}, {\em exact} methods for computing diverse sets
are largely missing with one notable exception: 
in \cite{DBLP:conf/icde/VeeSSBA08},
an efficient method for solving the $\diversityExplicit[\delta]$ problem
in a very specific setting is presented, where the diversity $\delta$ is defined as the sum over the ultrametric defined via an ordering of the attributes as recalled in Example~\ref{ex:tuplesUltraMetric}. 
Assuming the existence of a tree representation of the relation $S$ 
in the style of a Dewey tree known from 
XML query processing~\cite{DBLP:conf/sigmod/TatarinovVBSSZ02}, 
the algorithm in~\cite{DBLP:conf/icde/VeeSSBA08} finds a $k$-diverse set in $O(k)$ time.
Other than that, the field of tractable diversity computation is wide open and the main goal of this work is to fill this gap.
 	
	\section{Conclusions}\label{sec:conclusions}

In this work, we have studied the complexity of 3 levels of {\em diversity} problems. 
For the most basic problem $\diversityComputation[\delta]$ of computing the diversity $\delta(S)$ for a given set $S$ of elements, we have closed a problem left open in \cite{weitzman1992diversity} by 
proving intractability of this problem in case of the Weitzman diversity measure $\deltaW$.
We have then pinpointed the boundary between tractability and intractability for 
both, the $\diversityExplicit[\delta]$ and the $\diversityImplicit[\delta]$ problems
(i.e., the problems of maximizing the diversity of an explicitly or implicitly given 
set $S$, respectively) 
in terms of monotonicity properties of the diversity function $\delta$ extending 
an ultrametric.  
In particular, this has allowed us to identify tractable cases of 
the $\diversityImplicit[\delta]$ problem when considering acyclic conjunctive queries.

There are several natural directions of generalizing our results: clearly, they are  naturally extended to more general query classes than  
acyclic CQs such as
CQs with bounded (generalized or fractional) 
hypertree-width~\cite{DBLP:conf/pods/GottlobGLS16}. 
Less obvious is the extension of our tractability results to more general ultrametrics, which -- in addition to determining the first attribute (in a given order) {\em where} two tuples differ -- also introduce a measure 
{\em by how much} (again expressed as an ultrametric) the tuples differ in that attribute. 
Other directions of future work are concerned with studying other query languages over other data models such as (possibly restricted forms of) RPQs over graph data. 

\section*{Acknowledgments}
The work of Merkl and Pichler was supported by the Vienna Science and Technology Fund (WWTF) [10.47379/ICT2201]. The work of Arenas and Riveros was funded by ANID – Millennium Science Initiative Program – Code
ICN17\_0. Riveros was also funded by ANID Fondecyt Regular project 1230935. 
 
	\bibliographystyle{abbrv}
	\bibliography{biblio}
	
	\newpage
	\appendix
	\onecolumn

 	\section{Proofs of Section~\ref{sec:hardness}} 
    \label{app:hardnessWeitzman}

\subsection{Proof of Theorem~\ref{thm:weitzman:NP}}
	
\begin{proof}
Recall that we prove the 
$\mathsf{NP}$-hardness 
of $\diversityComputation[\deltaW]$
by reduction from the  \textsc{Independent Set} problem.
Let an arbitrary instance of \textsc{Independent Set} be given by a graph 
$G=(V, E)$ and integer $k$. Let  $|V| = n$. 
Then we set $S = V$ and 
$th = n - k + 2(k-1)$,
and we define the distance function $\dd$ on $S$ as follows: 
	$$
	\dd(u,v) = \begin{cases} 
		0 & \text{if $u = v$} \\  
		1 & \text{if $u$ and $v$ are adjacent in $G$} \\  
		2 & \text{otherwise}  
	\end{cases} 
	$$

 We first verify that $\dd$ satisfies the conditions of a metric: (1) $\dd(a,b) \geq 0$ for all $a,b \in S$ and
	$\dd(a,b) =  0$, if and only if $a = b$; 
	(2) $d$ is symmetric, i.e., $\dd(a,b) = \dd(b,a)$ for all $a,b \in S$; and 
	(3) for all pairwise distinct $a,b, \in S$, the triangle inequality holds, i.e., $\dd(a,c) \leq \dd(a,b) + \dd(b,c)$: indeed, this inequality holds, since we have $\dd(a,c) \leq 2$ and 
	$\dd(a,b), \dd(b,c) \geq 1$.
	It remains to show that $G$ has an independent set of size $k$, if and only if 
	$\deltaW(S) \geq th$. 
	
	For the ``only if'' direction, suppose that there exists an independent set $I\subseteq V$ of $G$ of size $k$.
	Let $S = V = \{v_1, v_2, \dots, v_n\}$ be an enumeration of the elements in $S$ such that the independent set 
	$I$ of $G$ is of the form $I = \{v_1,\dots, v_k\}$. 
	By Equation~(\ref{equ:deltaWS}), the following inequality holds: 
 \[
 \deltaW(S) \geq 
		   \sum_{i=k+1}^{n} \dd(v_{i}, \{v_{1}, \dots, v_{i-1} \}) + 
		\sum_{j=2}^{k}  \dd(v_{j}, \{v_{1}, \dots, v_{j-1} \})
	\]
Here, the first sum adds up the distances of $n-k$ elements $v_i$ from 
elements occurring before $v_i$ in this enumeration of $V = S$. 
Each such distance is at least 1 and, therefore, we have 
$\sum_{i=k+1}^{n} \dd(v_{i}, \{v_{1}, \dots, v_{i-1} \}) \geq n - k$.
For the second sum, recall that we are enumerating the elements in $V$ in such a way, that the first $k$ 
elements form an independent set of $G$.
Hence, by our definition of the distance function $\dd$, any vertex $v_j$ with $j \leq k$
has distance 2 from all other vertices in $\{v_{1}, \dots, v_{j-1} \}$. Hence, 
	$\sum_{j=2}^{k}  \dd(v_{j}, \{v_{1}, \dots, v_{j-1} \}) = 2(k-1)$ holds and, therefore, 
	$\deltaW(S) \geq n - k + 2(k-1) = th$ holds, as desired.
	
For the ``if''-direction, suppose that $\deltaW(S)\geq th =  n - k + 2(k-1)$ holds.
Let $\{v_1, \dots, v_n\}$ be an enumeration of the elements in $S = V$ that witnesses this diversity
according to Equation~(\ref{equ:deltaWS}), i.e., 
\[
		\deltaW(S)   = \ \sum_{i=2}^{n} \dd(v_{i}, \{v_{1}, \dots, v_{i-1} \})      
\]
By the definition of our distance function $d$, we have either $\dd(a,b) =1 $ or $\dd(a,b) = 2 $ for any two distinct elements $a,b \in S$. 
Hence, in order to get $\deltaW(S)\geq th =  n - k + 2(k-1)$,
there exist  (at least) $k-1$ elements 
	$v_{i_{k}}, \dots, v_{i_2}$, such that $i_{j}> i_\ell$ for $j > \ell$ and 
	$\dd(v_{i_j}, \{v_1,\dots, v_{i_{j}-1}\}) = 2.$
Moreover, we set $v_{i_1} := v_1$.
	Then $\dd(v_{i_j}, v_{i_\ell}) \geq \dd(v_{i_j}, \{v_1,\dots, v_{i_{j}-1}\}) = 2$ for all $j > \ell$. 
	Hence, by the definition of distance function $d$, 
	$v_{i_{k}}, \dots, v_{i_2}, v_{i_1}$ are pairwise non-adjacent in $G$, i.e., 
	they form an independent set of size $k$ in $G$.
\end{proof}

\subsection{Proof of Theorem~\ref{thm:weitzman:NPtuples}}
	
\begin{proof}
Suppose that an arbitrary instance of the \textsc{Independent Set} problem 
is given by a graph $G =(V,E)$ and integer $k$ with 
$V=\{v_1, \dots, v_n\}$ and $E = \{e_1, \dots, e_m\}$.
As was shown in \cite{DBLP:conf/ciac/AlimontiK97}, we may assume
that all vertices in $V$ have degree at most 3.

We define an instance of the $\diversityComputation[\deltaW]$ problem, 
where $S = \{t_1, \dots, t_n\}$ is a set of tuples and $th = 4 (n-k) + 5(k-1)$. 
The tuples in $S$ are constructed as follows (adapted from a construction used in \cite{MerklPS23}): We take as domain 
the set $\mathbb{D} = \{a_1, \dots, a_n, b_1, \dots, b_m\}$. 
The tuples in $S$  have arity 5, and 
we denote the $i$-th tuple $t_i$ in $S$
as $(t_{i,1}, \dots, t_{i,5})$. Then we construct the tuples $t_1, \dots, t_n$
by an iterative process as follows.

(1) First, for every $i \in \{1, \dots, n\}$, we initialize  $t_i = (a_i,a_i,a_i,a_i,a_i)$.
(2) Then we iterate through all edges $e_j \in E$ and do the following:  let $v_i$ and $v_{i'}$
be the endpoints of $e_j$. Let $k$ be an index, such that $t_{i,k} = a_i$ and $t_{i',k} = a_{i'}$, 
i.e., in the $k$-component, both $t_i$ and $t_{i'}$ still have their initial value. 
Then we set both $t_i$ and $t_{i'}$ to $e_j$.

Note that, in step (2) above, when processing an edge $e_j$, it is guaranteed that 
an index $k$ with $t_{i,k} = a_i$ and $t_{i',k} = a_{i'}$ exists. Here we
make use of the assumption that all vertices have degree at most 3. Hence, when processing 
edge $e_j$, each of the vertices $v_i$ and $v_{i'}$ has been considered at most twice in step (2) before.
Hence, in both $t_i$ and $t_{i'}$, at least 3 components still have the initial value 
$a_i$ and $a_{i'}$, respectively. By the pigeonhole principle, we may conclude that there
exists an index $k$, such that both $t_{i,k} = a_i$ and $t_{i',k} = a_{i'}$ still hold. 

As far as the Hamming distance $\dd(\cdot,\cdot)$ between any two distinct tuples 
$t_i,t_j$, is concerned, it is 
easy to verify that $\dd(t_i,t_j) = 4$ if $v_i,v_j$ are adjacent in $G$ and $\dd(t_i,t_j) = 5$  otherwise. 
Then the correctness proof of our problem reduction follows exactly the same lines as the proof of 
Theorem~\ref{thm:weitzman:NP} worked out above:
	
For the ``only if'' direction, suppose that there exists an independent set $I\subseteq V$ of $G$ of size $k$.
Let $V = \{v_1, v_2, \dots, v_n\}$ be an enumeration of the vertices in $V$ such that the independent set $I$ of $G$ is of the form $I = \{v_1,\dots, v_k\}$. 
By Equation~(\ref{equ:deltaWS}), the following inequality holds: 
 $$\deltaW(S) \geq 
		   \sum_{i=k+1}^{n} \dd(t_{i}, \{t_{1}, \dots, t_{i-1} \}) + 
		\sum_{j=2}^{k}  \dd(t_{j}, \{t_{1}, \dots, t_{j-1} \})$$
Again, the first sum adds up the distances of $n-k$ elements $t_i$ from 
elements occurring before $t_i$ in the enumeration of $S$ corresponding 
to the enumeration of $V$. 
Each such distance is at least 4 and, therefore, we have 
$\sum_{i=k+1}^{n} \dd(t_{i}, \{t_{1}, \dots, t_{i-1} \}) \geq 4(n - k)$.
For the second sum, recall that we are enumerating the elements in $V$ in such a way, that the first $k$ 
elements form an independent set of $G$.
Hence, as observed above, any tuple $t_j$ with $j \leq k$
has Hamming distance 5 from all other tuples in 
$\{t_{1}, \dots, t_{j-1} \}$. Hence, 
	$\sum_{j=2}^{k}  \dd(t_{j}, \{t_{1}, \dots, t_{j-1} \}) = 5(k-1)$ holds and, therefore, 
	$\deltaW(S) \geq 4(n - k) + 5(k-1) = th$ holds, as desired.
	
For the ``if''-direction, suppose that $\deltaW(S)\geq th =  4(n - k) + 5(k-1)$ holds.
Let $\{t_1, \dots, t_n\}$ be an enumeration of the elements in $S$ that witnesses this diversity
according to Equation~(\ref{equ:deltaWS}), i.e., 
$$
		\deltaW(S)   = \ \sum_{i=2}^{n} \dd(t_{i}, \{t_{1}, \dots, t_{i-1} \})      
$$
As observed above, for the Hamming distance $d$,
we have either $\dd(t,t') = 4 $ or $\dd(t,t') = 5 $ 
for any two distinct tuples $t,t' \in S$. 
Hence, in order to get $\deltaW(S)\geq th =  4(n - k) + 5(k-1)$,
there exist  (at least) $k-1$ elements 
	$t_{i_{k}}, \dots, t_{i_2}$, such that $i_{j}> i_\ell$ for $j > \ell$ and 
	$\dd(t_{i_j}, \{t_1,\dots, t_{i_{j}-1}\}) = 5.$
Moreover,  we set $t_{i_1} := t_1$.
	Then $\dd(t_{i_j}, t_{i_\ell}) \geq \dd(t_{i_j}, \{t_1,\dots, t_{i_{j}-1}\}) = 5$ for all $j > \ell$. 
	Hence, by the above observation on the Hamming distance between 
 tuples in $S$, 
	$v_{i_{k}}, \dots, v_{i_2}, v_{i_1}$ are pairwise non-adjacent in $G$, i.e., 
	they form an independent set of size $k$ in $G$.
\end{proof}

	\section{Proofs of Section~\ref{sec:ultrametric}} \label{app:ultrametric}

\subsection{Proof of Property~\ref{prop:ultrametric}}
\begin{proof}
We first prove Property (b): 
Suppose that $\ud(a_1,a_2) \leq r$. We have to show that then $\Ball(a_1, r) = \Ball(a_2, r)$ holds. We prove $\Ball(a_1, r) \subseteq \Ball(a_2, r)$. The other direction is symmetric.

Consider an arbitrary element $a \in \Ball(a_1, r)$. 
We have to show that also $a \in \Ball(a_2, r)$ holds.
By the strong triangle inequality, we have
$\ud(a_2,a) \leq \max\{\ud(a_2, a_1), \ud(a_1, a)\}$. Moreover, 
we are assuming $a \in \Ball(a_1, r)$. Hence, $\ud(a_1,a) \leq r$ holds. 
Together with $\ud(a_1,a_2) \leq r$ and by the symmetry property of metrics,
we thus have 
$\max\{\ud(a_2, a_1), \ud(a_1, a)\} \leq r$ and, therefore,
$\ud(a_2,a) \leq r$, i.e., $a \in \Ball(a_2, r)$.

To prove Property (a), let $\Ball(a_1, r_1) = B_1$ and $\Ball(a_2, r_2) = B_2$ for some $a_1, a_2 \in \U$ and $r_1, r_2 \in \bbQgeqz$. Moreover, suppose that $B_1 \cap B_2 \neq \emptyset$
holds. We have to show that then $B_1 \subseteq B_2$ or $B_2 \subseteq B_1$ holds.

Let $a \in B_1 \cap B_2$. Then, by definition, 
we have $\ud(a_1, a) \leq r_1$ and $\ud(a_2, a) \leq r_2$. 
By the strong triangle inequality, we have
$\ud(a_1,a_2) \leq \max\{\ud(a_1, a), \ud(a, a_2)\}$. 
Now suppose that $r_1 \leq r_2$ holds (the other case is symmetric). 
Then $\max\{\ud(a_1, a), \ud(a, a_2)\} \leq r_2$ and, therefore, 
also $\ud(a_1,a_2) \leq r_2$. 
By Property (b), we conclude that then $\Ball(a_1, r_2) = \Ball(a_2, r_2)$ holds. 
We are assuming $r_1 \leq r_2$. Hence, we have $\Ball(a_1, r_1) \subseteq \Ball(a_2, r_2)$.
\end{proof}

\subsection{Proof of Property~\ref{prop:radius}}
\begin{proof}
By contradiction, assume that there exists $a_1 \in B_1$ and $a_2 \in B_2$ such that $\ud(a_1, a_2) = r^* < \radius(B)$ (it cannot be bigger by the definition of $\radius(B)$). By Property~\ref{prop:ultrametric}(b), we know that $\Ball(a_1, r^*) = \Ball(a_2, r^*) = B^*  \in \Ball(\U)$ and $B^* \subsetneq B$. 
By $\tparent(B_1) = B$ it cannot happen that $B_1 \subsetneq B^*$, i.e., 
we cannot have $r_1 < r^* $. Instead, $r^* \leq r_1$ and, therefore, $B^* \subseteq B_1$ holds. 
Likewise, by $\tparent(B_2) = B$, it cannot happen that $B_2 \subsetneq B^*$, i.e., 
we cannot have $r_2 < r^* $. Instead, $r^* \leq r_2$ and, therefore, $B^* \subseteq B_2$ holds. 
Hence, by $B^* \subseteq B_1$ and $B^* \subseteq B_2$, we have 
$B^* \subseteq B_1 \cap B_2$ and, therefore,  
$B_1 \cap B_2 \neq \emptyset$, which contradicts the fact that $B_1$ and $B_2$ are disjoint (i.e., as children of $B$).   
\end{proof}

\subsection{Proof of Property~\ref{prop:build-tree}}

The construction of $\utree$ and $\radius(\cdot)$ works as follows.
First, we construct the graph $G=(S,S\times S)$ with edge weights $\ud$.
Second, we compute a minimum-weight spanning tree $\mathcal{T}$ of $G$.
Third, we initialize pointers $B$ from each $a\in S$ to $B(a)=\{a\}$.
Furthermore, for each $B(a)$ we store $C(B(a))=\emptyset$, $r(B(a))=0$.

Then, we iterate over the edges $uv$ in $\mathcal{T}$ in increasing order of the weight $\ud$ and do the following:
\begin{enumerate}
	\item Define $B_{new}=B(u)\cup B(v)$ and $r(B_{new})=\ud(u,v)$ and $C(B_{new}) = \emptyset$
	\item If $r(B(u))< \ud(u,v)$, add $B(u)$ to $\utree$ with children $C(B(u))$ and set $\radius(B(u)) = r(B(u))$.
	Furthermore, add $B(u)$ to $C(B_{new})$.
	\item Else, simply add $C(B(u))$ to $C(B_{new})$.
	\item If $r(B(v))< \ud(u,v)$, add $B(v)$ to $\utree$ with children $C(B(v))$.
	Furthermore, add $B(v)$ to $C(B_{new})$.
	\item Else, simply add $C(B(v))$ to $C(B_{new})$.
	\item Let the pointers of all $a\in B_{new}$ point to $B_{new}$, i.e., $B(a):=B_{new}$.
\end{enumerate}
In the end, we are left with the set $S$ every $a\in S$ points to, i.e, $B(a) = S$.
Therefore, finally, we also add $S$ to $\utree$ as the root with children $C(S)$ and set $\radius(\cdot) = r(S)$.

First note that this construction takes time $O(|S|^2)$ as there are only $O(|S|)$ many edges and each edge $uv$ is processed in linear time.

Now to the correctness.
Let $\mathcal{T}(a,r)$ be the subtree of $\mathcal{T}$ consisting of the nodes reachable from $a$ by using only edges with weight $\leq r$.
The algorithm iteratively computes these with $\mathcal{T}(u,r(B(u)))=B(u)$ whenever an edge $uv\in \mathcal{T}$ with $\ud(u,v)>r(B(u))$ is considered.
Furthermore, at that point, $C(B(u))$ are exactly those $\mathcal{T}(a,r)$ such that $a\in B(u)$, $r < r(B(u))$ but $r\leq r' < r(B(u))$ implies $\mathcal{T}(a,r) = \mathcal{T}(a,r')$.

We claim that $V(\mathcal{T}(a,r)) = \Ball_S(a,r)$ which would imply the correctness of the algorithm.
To that end, consider a $b\in \Ball_S(a,r)$.
In $G$, the edge $ab$ has weight $\ud(a,b) \leq r$ and thus, there must be a path from $a$ to $b$ in $\mathcal{T}$ only using edges with weights $\leq \ud(a,b)$ as $\mathcal{T}$ is a minimal-weight spanning tree.
Therefore, $b\in V(\mathcal{T}(a,r))$.

Conversely, consider a $b\in V(\mathcal{T}(a,r))$, i.e., there is a path $(a= p_1,\dots,p_l =b)$ in $\mathcal{T}$ with $\ud(p_i,p_{i+1}) \leq r$.
Thus, 
\[
\ud(a,b) \leq \max\{\ud (p_1,p_2), \ud (p_2,p_l)\}\leq \dots \leq \max\{\ud (p_1,p_2), \ud (p_2,p_3), \dots, \ud(p_{l-1},p_{l})\}\leq r
\]
and $b\in \Ball_S(a,r)$.

	\section{Proofs of Section~\ref{sec:explicit-rep}} \label{app:explicit-rep}

\subsection{Proof of Proposition \ref{propo:subsetMonotone}}

\begin{proof}	
	Let $A =\{a_1,\dots, a_t\} \subseteq \U,$ $B = \{b_1, \ldots, b_\ell\} \subseteq \U,$ and $B' = \{b_1', \ldots, b_\ell'\} \subseteq \U$ be such that $A \cap B = \emptyset = A \cap B'$, $\delta(B) \leq \delta(B'),$ and $\dd(a, b_i) = \dd(a, b_i')$ for every $a \in A$ and $i \leq \ell$.
	
	For $\delta =\deltasum$ we get
	\[
	\begin{array}{rcl}
		\deltasum(A \cup B) & = & \sum_{a, a' \in A} \dd(a, a') + \sum_{a \in A, b \in B} \dd(a,b) + \sum_{b,b'\in B} \dd(b,b')  \\
		& \leq & \sum_{a, a' \in A} \dd(a, a') + \sum_{a \in A, b \in B'} \dd(a,b) + \sum_{b,b'\in B'} \dd(b,b') \  = \  \deltasum(A \cup B').
	\end{array}
	\]
	
	For $\delta =\deltamin$ we get
	\[
	\begin{array}{rcl}
		\deltamin(A \cup B) & = & \min\{\min_{a\neq a' \in A} \dd(a, a'), \min_{a \in A, b \in B} \dd(a,b), \min_{b\neq b'\in B} \dd(b,b')\}  \\
		& \leq & \min\{\min_{a\neq a' \in A} \dd(a, a'), \min_{a \in A, b \in B'} \dd(a,b), \min_{b\neq b'\in B'} \dd(b,b')\} \\
		&  =  &  \deltamin(A \cup B').
	\end{array}
	\]
	
	For $\delta = \delta_W$ and an ultrametric distance function $\ud$ we get 
	\[
	\begin{array}{rcl}
		\deltaW(A \cup B) & = & \sum_{i=1}^t \min\{\ud(a_i,a_j),\ud(a_i,b)\mid i<j,b\in B\} + \sum_{i=1}^{l-1} \min\{\ud(b_i,b_j)\mid i<j\} \\
		& = & \sum_{i=1}^t \min\{\ud(a_i,a_j),\ud(a_i,b')\mid i<j,b'\in B'\} + \deltaW(B) \\
		& \leq & \sum_{i=1}^t \min\{\ud(a_i,a_j),\ud(a_i,b')\mid i<j,b'\in B'\} + \deltaW(B') \\
		& = & \sum_{i=1}^t \min\{\ud(a_i,a_j),\ud(a_i,b')\mid i<j,b'\in B'\} + \sum_{j=1}^{l-1} \min\{\ud(b_i,b_j)\mid i<j\} \\
		&  = & \deltamin(A \cup B').
	\end{array}
	\]	
	
	$\deltaW$ is not subset-monotone for arbitrary metrics $\dd$. 
	As an example consider $A = \{a,b_1,b_2,b_3\}, B=\{c_1,c_2,c_3,d_1,d_2\}, B'=\{c'_1,c'_2,c'_3,d'_1,d'_2\}$ with $\dd(a,b_i)=1=\dd(b_i,c_j)=\dd(c_i,d_j)=\dd(b_i,c'_j)=\dd(c'_i,d'_j)=\dd(c_i,c_j)=\dd(d'_1,d'_2)$. 
	All other distances equal $2$.
	Then, $\deltaW(A\cup B) = 12, \deltaW(A\cup B)' = 11, \deltaW(B)=5, \deltaW(B')=6$.
\end{proof}

\subsection{Proof of Theorem~\ref{theo:explicit-rep}}

\begin{proof}
	Let $S$ be a finite subset of the universe $\U$, $\ud$ an ultrametric over $\U$, and $\delta$ a subset-monotone diversity function of $\ud$. 
	By Proposition~\ref{prop:build-tree}, we can construct an ultrametric tree $\utree$ of $S$ in time $O(|S|^2)$. 
	For each vertex $B$ (i.e., a ball) of $\utree$, we maintain a function $\cand_B\colon \{0,\dots,\min\{k, |B|\}\} \rightarrow 2^{B}$ where $\cand_B(i)$ is a \emph{candidate} diverse subset of $B$ of size $i$. Formally, for every $i \in \{0,\dots,\min\{k, |B|\}\}$ we define:
	\[
	\cand_{B}(i) \ := \  \argmax_{A \subseteq B: |A| = i} \delta(A).
	\]
	Clearly, if we can compute $\cand_{S}$ for the root $S$ of $\utree$, then $\cand_{S}(k)$ is a diverse subset of $S$ of size $k$. 
	We can compute these functions $C_B$ in polynomial time using a dynamic programming approach.
	To that end, we compute $\cand_B$ for each $B \in \Ball_S$ in a bottom-up fashion over $\utree$. For every ball $B$, we can easily check that $\cand_B(0) = \emptyset$ and $\cand_B(1) = \{a\}$ for some $a \in B$. In particular, $\cand_{\{a\}} = \big\{0 \mapsto \emptyset, 1 \mapsto \{a\} \big\}$ is our base case for every leaf $\{a\}$ of $\utree$.
	For an inner vertex $B$ of $\utree$, the process is a bit more involved. 
	Let $B_1,\dots,B_l$ the children of $B$ in $\utree$ and assume that we already computed $\cand_{B_1},\dots, \cand_{B_l}$.
	
	We can now construct (see below) a binary tree $T_B$ with vertices $B_1,\dots, B_l, B_1^2:=\bigcup_{i=1}^2B_i, \dots, B_1^l:=\bigcup_{i=1}^lB_i=B$ and edges from $B_1^m=\bigcup_{i=1}^m B_i$ to $B_1^{m-1}=\bigcup_{i=1}^{m-1} B_i$ and $B_m$ for $1<m\leq l$.
	Thus, $B_1,\dots, B_l$ are the leaves and $B$ is the root of $T_B$.
	
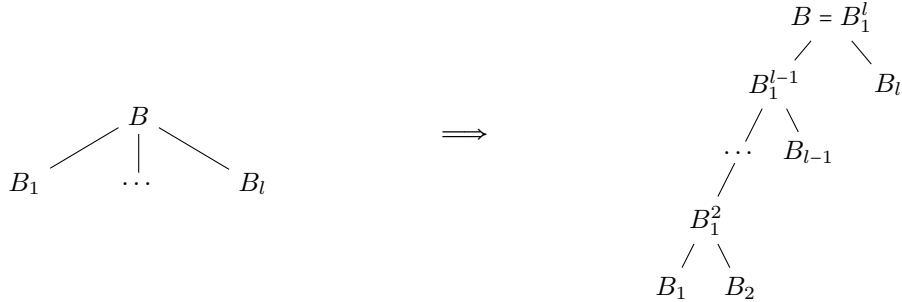
	\begin{minipage}{0.48\textwidth}
		\begin{center}
			\begin{tikzpicture}[level distance=.9cm,
				level 2/.style={sibling distance=0.7cm}]

				\node (t2) at (5,0) {$B$}
				child { node {$B_1$} }
				child { node {$\dots$} }
				child { node {$B_l$} };
			\end{tikzpicture}
		\end{center}
	\end{minipage}
	\begin{minipage}{0.04\textwidth}
		$\implies$
	\end{minipage}
	\begin{minipage}{0.48\textwidth}
		\begin{center}
			\begin{tikzpicture}[level distance=.9cm,
				level 2/.style={sibling distance=0.9cm}]

				\node (t2) at (5,0) {$B=B_1^l$}
				child { node {$B_1^{l-1}$} 
					child { node {$\dots$} 
						child { node {$B_{1}^2$}
							child { node {$B_1$} }
							child { node {$B_2$} }
						}
						child {edge from parent[draw=white]} 
					}
					child { node {$B_{l-1}$} } edge from parent[draw=black]
				}
				child { node {$B_l$} };
			\end{tikzpicture}
		\end{center}
	\end{minipage}

	To construct $\cand_B$, we first construct intermediate results $\cand_{B_1^m}$ with 
	\[
	\cand_{B_1^m}(i) \ := \  \argmax_{A \subseteq {B_1^m}: |A| = i} \delta(A).
	\]
	
	We claim that, for every $i \in \{0,\dots,\min\{k, |B_1^m|\}\}$, we can calculate $\cand_{B_1^m}(i)$ as $\cand_{B_1^m}(i) = \cand_{B_1^{m-1}}(i_1) \cup \cand_{B_m}(i_2)$
	where:
	\[
	(i_1,i_2) = \argmax_{(j_1,j_2): j_1+j_2 = i} \delta(\cand_{B_1^{m-1}}(j_1) \cup \cand_{B_m}(j_2)).
	\]
	
	To see that the claim holds, let $A$ be a subset of $B_1^m$ with $i$-elements maximizing $\delta(A)$. Define $A_1 := A \cap B_1^{m-1}$ and $A_2 := A \cap B_m$, and their sizes $i_1 := |A_1|$ and $i_2 := |A_2|$, respectively. Due to the correctness of $\cand_{B_1^{m-1}}$, we know that $\delta(A_1) \leq \delta(\cand_{B_1^{m-1}}(i_1))$. Further, $\ud(a_1, a_2) = \ud(a_1',a_2) = \radius(B)$ for every $a_1 \in A_1, a_1' \in \cand_{B_1^{m-1}}(i_1), a_2 \in A_2$ by Property~\ref{prop:radius}. Then the conditions of weak subset-monotonicity are satisfied and $\delta(A_1 \cup A_2) \leq \delta(\cand_{B_1^{m-1}}(i_1) \cup A_2)$. Following the same argument, we can conclude that $\delta(\cand_{B_1^{m-1}}(i_1) \cup A_2) \leq \delta(\cand_{B_1^{m-1}}(i_1) \cup \cand_{B_m}(i_2))$, proving that $\cand_{B_1^{m-1}}(i_1) \cup \cand_{B_m}(i_2)$ is optimal. 
	
	Thus, given $\cand_{B_1^{m-1}}$ and $\cand_{B_m}$, computing $\cand_{B_1^{m}}$ takes time $O(k^2 \cdot \fdelta(k))$.
	Consequently, given $\cand_{B_1},\dots, \cand_{B_l}$, computing $\cand_{B}$ requires time $O(k^2 \cdot \fdelta(k)\cdot l)$.
	In total, given the ultrametric tree $\utree$ and proceeding bottom-up, we can compute $\cand_{S}$ in time $O(k^2 \cdot \fdelta(k)\cdot |S|)$
\end{proof}

\subsection{Proof of Theorem~\ref{theo:exp-hardness}}

\begin{proof}
	We first give the intuition behind the reduction and then present if formally.
	The strategy is to reduce from \textsc{SAT}.
	To that end, for a given formula $\phi(\bar{x})$, we define a set of tuples $S_\phi$ such that for every possible truth assignment $\alpha\colon \bar{x} \rightarrow \{t,f\}$ there is a subset $S'\subseteq S_\phi$ that encodes $\alpha$.
	Moreover, these subsets will all be equal in size and Pareto optimal.
	Thus, a diversity function $\delta$ defined as 
	\begin{align*}
		\delta(S') = \begin{cases}
			\ud(e_1,e_2) & S' = \{e_1, e_2\}\\
			|S'| & S' \text{ represents } \alpha_{S'}  \text{ and } \alpha_{S'} \vDash \phi\\
			|S'|-1 & \text{otherwise}
		\end{cases}
	\end{align*}
	then is monotone and there is a $S'$ s.t., $\delta(S')=|S'|$ if and only if $\phi$ is satisfiable.
	Technical complication occur, however, as $\delta$ is not allowed to depend on $\phi$ and $\delta$ is only allowed to use the pairwise distances among $S'$ to determine a value of diversity. 
	
	Now to the formal construction.
	Let $\phi(x_1,\dots,x_n)$ be a propositional formula whose binary representation $\bin(\phi)=(\bin(\phi)_1,\dots,\bin(\phi)_m)$ requires, w.l.o.g.,  $||\phi||=m = m(n) := cn\log n$ bits for some constant $c$.
	Note that $m\colon \mathbb{N}\rightarrow \mathbb{N}$ is an injection.
	The set of elements $S_\phi$ shall be tuples over $\{\phi, C_1,\dots, C_{m(3n+2m+2)}, B,1,0_1,\dots, 0_{m+1}, X,T,F_1,\dots,F_{n+1},D\}$.
	Concretely, $S$ (we omit $\phi$ from $S_\phi$ if $\phi$ is clear from the context) consists of the following elements $S=S^V \cup S^T \cup S^F \cup S^B \cup S^1 \cup S^0 \cup S^C$
	{\small
	\begin{alignat*}{10}
		S^V=\{&t^V_i=&&(\underbrace{\phi, \dots, \phi}_{m(2m+2n+3)}; V;&& \underbrace{D,\dots,D}_{2m+1};&&  \underbrace{D, \dots, D}_{i-1}, V,\underbrace{D, \dots, D}_{2n-2i+1}, V, \underbrace{D, \dots, D}_{i-1})&& \\
		&&&&&&&\quad\mid i=1,\dots, n\},\\
		S^T=\{&t^T_i=&&( \underbrace{\phi, \dots, \phi}_{m(2m+2n+3)}; V;&& \underbrace{D,\dots,D}_{2m+1};&&  \underbrace{D, \dots, D}_{i-1}, V,\underbrace{D, \dots, D}_{2n-2i+1}, T, \underbrace{D, \dots, D}_{i-1})&&\\
		&&&&&&&\quad\mid i=1,\dots, n\},\\
		S^F=\{&t^F_i=&&( \underbrace{\phi, \dots, \phi}_{m(2m+2n+3)}; V;&& \underbrace{D,\dots,D}_{2m+1};&&  \underbrace{D, \dots, D}_n, F_i, \underbrace{D, \dots, D}_n) &&\\
		&&&&&&&\quad\mid i=1,\dots, n+1\},\\
		S^B=\{&t^B_i=&&( \underbrace{\phi, \dots, \phi}_{m(2m+2n+3)}; B;&& \underbrace{D, \dots, D}_{i-1}, B,\underbrace{D, \dots, D}_{2m-2i+1}, B, \underbrace{D, \dots, D}_{i-1};&& \underbrace{D,\dots,D}_{2n+1})&& \\
		&&&&&&&\quad\mid i=1,\dots, m\},\\
		S^1=\{&t^1_i=&&( \underbrace{\phi, \dots, \phi}_{m(2m+2n+3)}; B;&& \underbrace{D, \dots, D}_{i-1}, B,\underbrace{D, \dots, D}_{2m-2i+1}, 1, \underbrace{D, \dots, D}_{i-1};&& \underbrace{D,\dots,D}_{2n+1}) &&\\
		&&&&&&&\quad\mid i=1,\dots, m \text{ and } 1 = \bin(\phi)_i\},\\
		S^0=\{&t^0_i=&&( \underbrace{\phi, \dots, \phi}_{m(2m+2n+3)}; B;&& \underbrace{D, \dots, D}_m, 0_i, \underbrace{D, \dots, D}_m;&& \underbrace{D,\dots,D}_{2n+1}) &&\\
		&&&&&&&\quad\mid i=1,\dots, m \text{ and } 0 = \bin(\phi)_i\} \\
		\cup \{&t^0_{m+1}=&&( \underbrace{\phi, \dots, \phi}_{m(2m+2n+3)}; B;&& \underbrace{D, \dots, D}_m, 0_{m+1}, \underbrace{D, \dots, D}_m;&& \underbrace{D,\dots,D}_{2n+1}) \}.&&\\
		S^C=\{& &&(\underbrace{\phi, \dots, \phi}_{m(2m+2n+3)}; C_i;&&\underbrace{D,\dots, D}_{2m+1};&&\underbrace{D,\dots, D}_{2n+1}) &&\\
		&&&&&&&\quad\mid i= 1, \dots, m(3n+2m+2) \}.
	\end{alignat*}
	}	

	The whole universe $\U$ consists of tuples that appear in some $S_\phi$ for some formula $\phi$.
	Each $S_\phi$ contain $(m_{\phi}+1)(3n_{\phi}+2m_{\phi}+2)$ elements where $||\phi|| = m_{\phi}$ and $n_{\phi} = m^{-1}(m_{\phi})$.
	
	The distance $\ud$ between two elements $e,e'$ is $\ud(e,e') := 2^{-i}$ where $i$ is the last position $e,e'$ agree on (and $\ud(e,e')=0$ if $e=e'$). 
	One can easily check that $\ud$ is an ultrametric.
	Thus, for elements $e\in S_{\phi}$ and $e'\in S_{\phi'}$, their distance is $\ud(e,e') := 1$ while for elements $e,e'\in S_{\phi}$, their distance is $\ud(e,e') \leq 2^{-m_{\phi}(2m_{\phi}+2n_{\phi}+3)}$ due to the prefix.
	Similarly, the distance between elements $e,e'\in S^V_{\phi}\cup S^T_{\phi}\cup S^F_{\phi}$< is $\ud(e,e') \leq 2^{-m_{\phi}(2m_{\phi}+2n_{\phi}+3)}2^{-2m_{\phi}-2}$, and so on $\dots$
		
	Now, given a set $S'\subseteq \U$, we classify $S'$ as \textit{erroneous}, if there are two tuples $(\phi, \dots), (\phi', \dots)\in S'$ such that $\phi \neq \phi'$.
	Thus, if $S'$ is not erroneous, $S'\subseteq S_\phi$ for some formula $\phi$.
	We say $S'$ is \textit{small} if $|S'| \leq m(3n+2m+2)$.
	Then, we define an isomorphism $f$ on $S'$ to some subset of $S_\phi$.
	Concretely,
	\begin{align*}
		f(t) = \begin{cases}
			t^V_i & t = t^T_i \text{ and } t^V_i\not\in S'\\
			t & \text{ otherwise}
		\end{cases}
	\end{align*}
	We say $S'\subseteq S$ of size $2m+2n+2 + m(3n+2m+2)$ is valid if $S^V\cup S^B\cup S^1\cup S^0 \cup S^C\subseteq f(S')$.
	For valid sets $S'\subseteq S$ we define the variable assignments $\alpha_{S'}\colon \bar{x} \rightarrow \{t,f\}$ where $\alpha(x_i) = t$ iff $t^T_i \in f(S')$.
	
	Now consider the following function $\delta\colon \setsU \rightarrow \bbQgeqz$:
	\begin{align*}
		\delta(S') = \begin{cases}
			\ud(e_1,e_2) & S' = \{e_1, e_2\}\\
			|S'| & S' \text{ is erroneous }\\
			|S'|-1 & S' \text{ is small } \\
			|S'| & S'\subseteq S_\phi \text{ is valid and } \alpha_{S'} \vDash \phi\\
			|S'|-1 & \text{otherwise}
		\end{cases}
	\end{align*}

	$\delta$ is well-defined and, for a given formula $\phi$, there is a $S'\subseteq S_\phi$ s.t., $\delta(S')=|S'|=2m+2n+2+ m(3n+2m+2)$ if and only if $\phi$ is satisfiable.
	
	\paragraph{Correctness}
	It remains to show that $\delta$ is monotone (this then also implies closedness under isomorphisms).
	For non-monotonicity, there have to be sets $A=\{a_1,\dots, a_l\},B=\{b_1,\dots,b_l\}\subseteq \U$ such that $\dd(a_i,a_j)\leq \dd(b_i,b_j)$ for all $1\leq i,j\leq l$ but $\delta(B)<\delta(A)$.
	Let $h$ be the function $h(a_i)=b_i$.
	Firstly, $B$ must at least contain 3 elements, thus, $2\leq \delta(B)=|B|-1 < |A|=\delta(A)$.
	Consequently, $B$ is not erroneous, i.e., $B\subseteq S_\phi$ for some formula $\phi$ and, importantly, $\dd(b_i,b_j)<1$.
	This implies that also all tuples of $A$ must start with the same formula $\phi'$ as $\dd(a_i,a_j)\leq \dd(b_i,b_j)<1$, i.e., $A\subseteq S_{\phi'}$ for some formula $\phi'$.
	
	Let $m_\phi:= ||\phi||, m_{\phi'}:= ||\phi'||, n_{\phi} := m^{-1}(m_{\phi}),$ and $n_{\phi'} := m^{-1}(m_{\phi'})$.
	Due to the length of  elements in $A$, the minimal distances between $a_i,a_j$ is $2^{-(m_{\phi'}+1)(2m_{\phi'}+2n_{\phi'}+3)+1}$.
	Likewise, due the prefixes $(\phi, \dots, \phi)$ of elements in $B$ of length $m_{\phi}(2m_{\phi}+2n_{\phi}+3)$, we have a bound on the maximal distances between $b_i,b_j$.
	Consequently, $2^{-(m_{\phi'}+1)(2m_{\phi'}+2n_{\phi'}+3)+1} \leq \dd(a_i,a_j) \leq \dd(b_i,b_j) \leq 2^{-m_{\phi}(2m_{\phi}+2n_{\phi}+3)}$, which implies $m_{\phi'}\geq m_{\phi}$.
	Thus, if $B$ would be small, $A$ would also be small.
	But $A$ cannot be small as $\delta(A)=|A|$ and, therefore, neither is small.
	
	Consequently, $A$ must be valid, and $S^B_{\phi'} \cup S^1_{\phi'} \cup S^0_{\phi'} \cup S^C_{\phi'} \subseteq A \subseteq S_{\phi'}$.
	Furthermore, $l=2m_{\phi'}+2n_{\phi'}+2 + m_{\phi'}(3n_{\phi'}+2m_{\phi'}+2)$ and $(m_{\phi}+1)(3n_{\phi}+2m_{\phi}+2) =|S_\phi|\geq |B| = l$ which implies $m_\phi \geq m_{\phi'}$.
	Hence, $m_\phi = m_{\phi'}=:m, n_\phi = n_{\phi'}=:n, |S_\phi| = |S_{\phi'}|$.

	Now consider elements $(\phi',\dots, \phi';C_i)\in S^C_{\phi'}\subseteq A$.
	The distance between $(\phi',\dots, \phi';C_i)$ and any element $a_j\neq (\phi',\dots, \phi';C_i), a_j\in A$ is $2^{-m(2m+2n+3)}$.
	Thus, this must also be the case for elements $h(\phi',\dots, \phi';C_i)$, i.e., $h(S^C_{\phi'}) = S^C_{\phi}\subseteq B$.
	
	Now consider elements $t^{(V/T/F)}_{i,\phi'}=(\phi',\dots, \phi';V;\dots)\in (S^V_{\phi'} \cup S^T_{\phi'} \cup S^F_{\phi'}) \cap A \subseteq A$ and $t^{(B/1/0)}_{j,\phi'}=(\phi',\dots, \phi';B;\dots)\in S^B_{\phi'} \cup S^1_{\phi'} \cup S^0_{\phi'} \subseteq A$.
	The distance between $t^{(V/T/F)}_{i,\phi'}$ and any element $t^{(B/1/0)}_{j,\phi'}$ is $2^{-m(2m+2n+3)}$ while the distance among elements $t^{(V/T/F)}_{i,\phi'}$ and among elements $t^{(B/1/0)}_{j,\phi'}$ is strictly less than $2^{-m(2m+2n+3)}$.
	Furthermore, since $n_{\phi'} \ll m_{\phi'}$, also $|(S^V_{\phi'} \cup S^T_{\phi'} \cup S^F_{\phi'}) \cap A| \ll |S^B_{\phi'} \cup S^1_{\phi'} \cup S^0_{\phi'}|$.
	Combining these arguments, we assert that $h((S^V_{\phi'} \cup S^T_{\phi'} \cup S^F_{\phi'}) \cap A) = (S^V_{\phi} \cup S^T_{\phi} \cup S^F_{\phi}) \cap B$ and $h(S^B_{\phi'} \cup S^1_{\phi'} \cup S^0_{\phi'}) = S^B_{\phi} \cup S^1_{\phi} \cup S^0_{\phi}$.

	First notice that we can assume, w.l.o.g., $t^B_{i,\phi}\in B$ if $t^1_{i,\phi}\in B$ due to both elements being indistinguishable.
	Furthermore, we can assume that $h^{-1}(t^B_{i,\phi}) \in S^{B}_{\phi'}$ if $h^{-1}(t^1_{i,\phi}) \in S^{B}_{\phi'}$ again due to both elements being indistinguishable.
	Now consider elements $\{t^B_{1,\phi'},\dots, t^B_{m,\phi'}\} = S^B_{\phi'}$ and $t^0_{1,\phi'}$ with distances
	\begin{align*}
		\dd(t^B_{i,\phi'}, t^B_{j,\phi'}) &= \begin{cases}
			2^{-m(2m+2n+3)}2^{-\min\{i,j\}} & i\neq j\\
			0 & i = j
		\end{cases}\\
		\dd(t^B_{i,\phi'}, t^0_{1, \phi'}) &= 2^{-m(2m+2n+3)}2^{-i} 
	\end{align*}
	Thus, we also have to achieve at least these distances among $h(S^B_{\phi'} \cup \{t^0_{1,\phi'}\})$.
	This is only possible (with the previous assumptions) if $h(t^B_{i,\phi'}) = t^B_{i,\phi}$ and $h(t^0_{1,\phi'}) \in S^0_\phi$.
	Furthermore, the distances of the remaining elements $t^0_{i,\phi'}\in S^0_{\phi'}, i \neq 1$ to $S^B_{\phi'}$ also imply that $h(t^0_{i,\phi'})\in  S^0_\phi$.
	Importantly, thus, $|S^0_{\phi'}|=|S^0_{\phi}|$.
	Moreover, the distances of elements $t^1_{i,\phi'}\in S^1_{\phi'}$ to $S^B_{\phi'}$ and $t^0_{1,\phi'}$ are
	\begin{align*}
		\dd(t^1_{i,\phi'}, t^B_{j,\phi'}) &= \begin{cases}
			2^{-m(2m+2n+3)}2^{-\min\{i,j\}} & i\neq j\\
			2^{-m(2m+2n+3)}2^{-(2m-i+2)} & i = j
		\end{cases}\\
		\dd(t^1_{i,\phi'}, t^0_{1, \phi'}) &= 2^{-m(2m+2n+3)}2^{-i} 
	\end{align*}
	Thus, also $h(t^1_{i,\phi'}) = t^1_{i,\phi}$.
	Importantly, thus, $S^1_{\phi} = \{t^1_{i,\phi} \mid t^1_{i,\phi'}\in S^1_{\phi}\}$.
	However, 
	\[
	\bin(\phi) = (\bin(\phi)_1,\dots, \bin(\phi)_m)
	\] 
	and $\bin(\phi)_i = 1$ if and only if $t^1_{i,\phi}\in S^1_{\phi}$.
	Thus, 
	\[
	\bin(\phi) = (\bin(\phi)_1,\dots, \bin(\phi)_m) = (\bin(\phi')_1,\dots, \bin(\phi')_m) = \bin(\phi'),
	\] 
	That is, $\phi = \phi'$.
	
	An analogous argument can be used to show that $(S^V_{\phi} \cup S^T_{\phi}) \cap A = (S^V_{\phi} \cup S^T_{\phi}) \cap B$ and $|S^F_{\phi}\cap A| = |S^F_{\phi}\cap B|$.
	We would have to proceed as follows:
	First, as $A$ is valid, $S^V_{\phi}\subseteq A$ and some $t^F_{i}\in A$.
	Thus, w.l.o.g., $h(S^V_{\phi}) = S^V_{\phi}$ and $h(t^F_{i})\in S^F$.
	Then, we can assert $h(S^F_{\phi}\cap A) = S^F_{\phi}\cap B$ and $h(S^T_{\phi}) = S^T_{\phi}$.
	
	Now, given $(S^V_{\phi} \cup S^T_{\phi}) \cap A = (S^V_{\phi} \cup S^T_{\phi}) \cap B$ and $|S^F_{\phi}\cap A| = |S^F_{\phi}\cap B|$, we can conclude that $B$ is in fact valid as $\alpha_B = \alpha_A$, and $\delta(B) = |B| = |A| = \delta(A)$.
	This completes the proof.	
\end{proof} 
	
	\section{Proofs of Section~\ref{sec:implicit-rep}} \label{app:implicit-rep}

\subsection*{Proof of Theorem~\ref{theo:ptime-implicit-rep}}

\newcommand{\mroot}{B_\texttt{root}}
\newcommand{\mnew}{B_\texttt{new}}
\newcommand{\mparent}{B_\texttt{parent}}
\newcommand{\irinit}{\texttt{init}}
\newcommand{\irnext}{\texttt{next}}
\newcommand{\ircurrent}{\texttt{current}}

\begin{algorithm}[t]
	\DontPrintSemicolon
	\KwIn{An instance $I \in \irI$ and $k \in \bbN$.}
	\KwOut{A $k$-diversity set $S \subseteq \sem{I}$ with respect to $\delta$.}
	$\mroot \gets \irroot(I)$ \\
	$S \gets \{\irmember(I, \mroot)\}$ \\
	$L \gets \{\mroot\}$ \\
	$\irchildren(I, \mroot).\irinit$ \\ 
	$\irchildren(I, \mroot).\irnext$ \\
	\While{$|S| < k \, \wedge \, L \neq \emptyset$ \ }{
		$B \gets \argmax_{B \in L} \delta(S \cup \{\irmember(I, \texttt{Children}(I,B).\ircurrent)\}$ \\
		$S \gets S \cup \{\irmember(I, \texttt{Children}(I,B).\ircurrent)\}$\\
		\If{$\irchildren(I,B).\irnext = \texttt{false}$} {
			$L \leftarrow L \setminus \{B\}$ \\
			\For{$B' \in \irchildren(I,B)$}{
				\If{$|B'| > 1$} {
					$L \leftarrow L \cup \{B'\}$ \\
					$\irchildren(I, B').\irinit$ \\
					$\irchildren(I, B').\irnext$ \\
				}
			}
		}
	}
	\Return{$S$}
	\caption{For fixed ultrametric $\ud$ and implicit representation $(\irI, \sem{\cdot})$ over a common $\U$, implicit ultrametric tree $(\irroot, \irchildren, \irmember)$, and  subset-monotone diversity function $\delta$ of~$\ud$, compute, for an instance $I\in \irI$, a $k$-diverse subset of $\sem{I}$.}
	\label{alg:divAnsImpSSM}
\end{algorithm}

\begin{proof}
	In Algorithm~\ref{alg:divAnsImpSSM}, we provide all the instructions for the algorithm of  Theorem~\ref{theo:ptime-implicit-rep}. This algorithm assumes a fixed implicit representation $(\irI, \sem{\cdot})$ over $\U$, including a fixed implicit ultrametric tree given by the methods $(\irroot, \irchildren, \irmember)$. In addition, the ultrametric $\ud$ over $\U$ and the subset-monotone diversity function $\delta$ of~$\ud$ are fixed. Then, given an instance $I \in \irI$ and a $k \in \bbN$, the algorithm computes a $k$-diverse set $S \subseteq \sem{I}$ with respect to $\delta$.
	
	Recall that, given a ball $B \in \Ball(\sem{I})$, the method $\irchildren(I, B)$ enumerates the children of $B$ in $\utree[\sem{I}]$ with polynomial delay. For using this enumeration procedure, we assume an iterator interface with methods $\irinit$, $\irnext$, and $\ircurrent$, such that:
	\begin{itemize}
		\item $\irchildren(I, B).\irinit$ starts the iteration, placing the $\ircurrent$ pointer to the first child of $B$;
		\item $\irchildren(I, B).\ircurrent$ retrieves the current child of $B$;
		\item $\irchildren(I,B).\irnext$ moves the $\ircurrent$ pointer to the next child of $B$, outputting $\texttt{true}$ if a next child exists, and $\texttt{false}$, otherwise.  
	\end{itemize}
	Note that $\irchildren(I, B).\ircurrent$ retrieves the same child of $B$ whenever we call it consecutively (i.e., without calling $\irchildren(I, B).\irinit$ or $\irchildren(I,B).\irnext$ in between). By the definition of an implicit ultrametric tree, the running times of methods $\irinit$, $\ircurrent$, and $\irnext$ are bounded by $O(\fT(I))$. 
	
	For the sake of simplification, we assume that the method $\irmember(I, B)$ always outputs the same solution as $\irmember(I, B_{\text{fc}})$ where $B_{\text{fc}}$ is the first-child of $B$ in the implicit ultrametric tree. In other words, if we call $\irchildren(I, B).\irinit$, then it always holds that:
	\[
	\irmember(I, B) = \irmember(I, \irchildren(I, B).\ircurrent). \tag{\ddagger}
	\] 
	This assumption simplifies the presentation of the algorithm considerably, and we can assume it without loss of generality. Indeed, if it does not hold, the algorithm will require an extra check, increasing the whole running time only by a constant. 
	
	Intuitively, Algorithm~\ref{alg:divAnsImpSSM} keeps a set $S$ of solutions, and its primary goal is to maximize the ``incremental'' diversity of adding a new element to $S$. This new element is chosen from the current balls of the ultrametric tree that have not been visited yet. For this purpose, the algorithm maintains a set $L$ of balls such that, for each $B \in L$, the algorithm is iterating through the children of $B$ by using the iteration interface $\irchildren(I, B)$. The algorithm picks the new element that maximizes the incremental diversity of $S$ greedily from $L$ by computing (line~7):
	\[
	B\gets \argmax_{B \in L} \delta(S \cup \{\irmember(I, \texttt{Children}(I,B).\ircurrent)\}.
	\]
	Then the algorithm adds $\irmember(I, \texttt{Children}(I,B).\ircurrent)$ to $S$ (line 8) and moves to the next children of $B$ (line 9) until $S$ has size $k$ or $L$ is empty (line 6).
	
	Following the above strategy, Algorithm~\ref{alg:divAnsImpSSM} starts by picking the root ball $\mroot = \irroot(I)$ and adds a first solution to $S$ (line 2). Then, it adds $\mroot$ to $L$ for iterating through its children (lines 3 and 4). By assumption $(\ddagger)$, the solution in $\mroot$ added to $S$ is the same as the first children of $\mroot$. For this reason, the algorithm skips the first children of $\mroot$ by calling $\irchildren(I, \mroot).\irnext$. 
	
	When Algorithm~\ref{alg:divAnsImpSSM} reaches the end of the children of a ball $B \in L$ (line 9), it removes $B$ from $L$ (line 10). After this, it iterates over all the children $B'$ of $B$ (line 11), adding $B'$ to $L$ whenever $B'$ is not a leaf node in the ultrametric tree, namely, $|B'| > 1$ (lines 12 and 13). Similar to the initial case, by assumption $(\ddagger)$ we can skip the first sibling $B_{\text{fc}}$ of $B'$, given that the solution that $B_{\text{fc}}$ can contribute to $S$ is already in $S$ (lines 14 and 15).
	
	\paragraph{Correctness} We prove the correctness of Algorithm~\ref{alg:divAnsImpSSM} in three steps.
	
	For the first step, we prove that, at the beginning of each iteration (line 6), $S$ and $L$ cover all solutions $\sem{I}$, i.e., $S\cup \bigcup_{B\in L} B = \sem{I}$.
	We show this by induction on the number of iterations. 
	This is certainly true before the first iteration, since $\mroot$ is the root of the ultrametric tree and $\mroot = \sem{I}$. 
	For any iteration, $S$ only grows while $L$ only changes if  $\irchildren(I,B).\irnext = \texttt{false}$ (line 9).
	Then, the algorithm removes $B$ from $L$ and adds all $B' \in \irchildren(I,B)$ to $L$ for which $|B'| > 1$. 
	Note that, if $|B'| = 1$, then $B' \subseteq S$ as at least one element $a\in B'$ of each child $B'$ of $B$ was added to $S$.
	Given that $\irchildren(I,B)$ forms a partition of $B$, we can assert that $B\subseteq S\cup\bigcup_{B'\in \irchildren(I,B), |B'|>1} B'$. 
	We conclude that $S\cup \bigcup_{B\in L} B = \sem{I}$ still holds.

	For the second step, we show that, at the beginning of each iteration, for any $a \in \sem{I} \setminus S$, there exists $B \in L$ such that $\delta(S\cup \{a\}) \leq \delta(S \cup \{b\})$ for $b = \irmember(I, \texttt{Children}(I,B).\ircurrent)$. In other words, in steps 7 and 8 the algorithm chooses an element that maximizes the incremental diversity of $S$. 
	To prove this, take any $a \in \sem{I} \setminus S$.
	By the first step, there must exists $B \in L$ such that $a \in B$. Let $B' =  \texttt{Children}(I,B).\ircurrent$ and $b = \irmember(I,B')$. 
	On the one hand, for every $s \in S \cap B$, it holds that  $\ud(s, a) \leq r(B)$. Given that $S \cap B' = \emptyset$ (i.e., $B'$ has not been considered yet), $\ud(s, b) = r(B)$ (by Property~\ref{prop:radius}). In particular, $\ud(s, a) \leq \ud(s, b)$.
	On the other hand, for every $s \in S \setminus B$, it holds that $\ud(s, a) = \ud(s, b)$ given that $a, b \in B$. Combining both cases, we have that $\ud(s, a) \leq \ud(s, b)$ for every~$s \in S$.
	Now, if we choose $A = S$, $B = \{a\}$, and $B' = \{b\}$, we conclude by subset-monotonicity that $\delta(S \cup \{a\}) \leq \delta(S \cup \{b\})$.
	
	For the last step, we prove that the algorithm always outputs a $k$-diverse set with respect to $\delta$. 
	Let $S = \{s_1, \ldots, s_k\}$ be the output of the algorithm where $s_1, \ldots, s_k$ is the order how the algorithm added the elements to $S$. 
	Towards a contradiction, assume that there exists $S' \subseteq \sem{I}$ of size $k$ such that $\delta(S) < \delta(S')$.  Let $M$ be the set of all $S'$ such that $\delta(S) < \delta(S')$ and $|S'| = k$. Pick one $S^* \in M$ that contains the longest prefix of $s_1,\ldots, s_k$, namely, $S^* = \argmax_{S' \in M} \{m \mid s_1, \ldots, s_m \in S'\}$. Then $s_1, \ldots, s_m \in S^*$ but $s_{m+1} \notin S^*$. 
	Also, let $s^* \in S^*$ be such that $\ud(s_{m+1}, S^* \setminus \{s_1, \ldots, s_m\}) = \ud(s_{m+1}, s^*)$ (i.e., $s^*$ is one of the closest elements to $s_{m+1}$ in $S^*\setminus \{s_1, \ldots, s_m\}$). Define $A = S^* \setminus \{s_1, \ldots, s_m, s^*\}$, $B =  \{s_1, \ldots, s_m, s^*\}$, and $B' = \{s_1, \ldots, s_m, s_{m+1}\}$.
	By the second step, we know that $\delta(B) \leq \delta(B')$. Furthermore, for every $a \in A$ we have that:
	\[
	\ud(a, s^*) \leq \max\{\ud(a,s_{m+1}), \ud(s_{m+1}, s^*)\} = \max\{\ud(a, s_{m+1}), \ud(s_{m+1},S^* \setminus \{s_1, \ldots, s_m\})\} = \ud(a, s_{m+1}).
	\]
	The remaining elements $B\setminus\{s^*\}$ are the same as $B'\setminus \{s_{m+1}\}$ and $\ud(a,s_i)=\ud(a,s_i)$.
	Then, applying subset-monotonicity, we get that:
	\[
	\delta(S^*) = \delta(A \cup B) \leq \delta(A \cup B')
	\]
	This means that $A \cup B' \in M$ but $A \cup B'$ has a longer prefix of $s_1,\ldots, s_k$ than $S^*$, which is a contradiction. We conclude that the output $S$ of Algorithm~\ref{alg:divAnsImpSSM} is a $k$-diverse set of $\sem{I}$ with respect to $\delta$. 
	
	\paragraph{Running time} By inspecting  Algorithm~\ref{alg:divAnsImpSSM}, one can check that the number of balls of the ultrametric tree that are used is at most $O(k)$. By caching the result of the functions $\irchildren$ and $\irmember$, we can bound the running time of these methods in the algorithm by $O(k \cdot \fT(I))$. 
	
	Now, let $\fdelta(n)$ be a function such that the running time of computing $\delta(A)$ over a set $A$ of size $n$ is bounded asymptotically by $\fdelta(n)$. One can check that $|L| \in O(k)$ and, for each iteration, we compute $B$ by calling $\delta$ exactly $|L|$-times over a set of size at most $k$. Then, the running time of line~7 takes at most $O(k \cdot \fdelta(k))$. Overall, we can bound the running time of Algorithm~\ref{alg:divAnsImpSSM} by:
	\[
	O(k \cdot \fT(I) + k^2 \cdot \fdelta(k)).
	\]
\end{proof}

\subsection{Proof of Proposition~\ref{prop:sum-min}}

\begin{proof}
	We show that $\deltasummin$ is subset monotone.
	Let $A\subseteq \U,B=\{b_{1},\dots,b_l\},B' = \{b'_1,\dots,b'_l\}$ with $\emptyset = A\cap B = A\cap B'$ be such that $\deltasummin(B)\leq \deltasummin(B')$ and $\ud(a,b_i) = \ud(a,b'_i)$ for any $a\in A$ and $ i\in \{1,\dots, l\}$.
	We need to show that $\deltasummin(A \cup B)\leq \deltasummin(A \cup B')$.
	To that end, we start by computing the contributions of $a\in A$, i.e.,
	\begin{align*}
		\ud(a,A\cup B \setminus \{a\}) &= \min \{ \min_{a'\in A:a'\neq a}\ud(a,a'),  \min_{b_i\in B} \ud(a,b_i) \} \\
		&= \min \{ \min_{a'\in A:a'\neq a}\ud(a,a'), \min_{b'_i\in B'} \ud(a,b'_i) \} \\
		&= \ud(a,A\cup B' \setminus \{a\}).
	\end{align*}
	
	For the elements $b_i\in B$ the argument is more involved.
	For that, let 
	\begin{align*}
		B(b_i) &:= \{b_j \in B\mid \ud(b_i,b_j) \leq \min \{ \ud(b_i,A), \ud(b_j,A)\}\}\subseteq B \\
		B'(b'_i)&:=  \{b'_j \in B'\mid \ud(b'_i,b'_j) \leq \min \{ \ud(b'_i,A), \ud(b'_j,A)\}\}\subseteq B'.
	\end{align*}
	Of course, $B(b_i) = B(b_j)$ for $b_j\in B(b_i)$.
	Furthermore, for $b_i,b_j \in B(b_i)$ there is a $a(b_i,b_j)\in A$ such that $\ud(b_i,b_j) \leq \min\{\ud(b_i,A), \ud(b_j,A)\} = \min \{ \ud(b_i,a(b_i,b_j)), \ud(b_j,a(b_i,b_j))\} $.
	Thus, also $\ud(b_i,a(b_i,b_j))= \ud(b_j,a(b_i,b_j))$.
	Then, also 
	\begin{align*}
		\ud(b'_i,b'_j) &\leq \max \{\ud(b'_i,a(b_i,b_j)), \ud(b'_j,a(b_i,b_j))\}\\
		&= \max \{\ud(b_i,a(b_i,b_j)), \ud(b_j,a(b_i,b_j))\}\\
		&= \ud(b_i,a(b_i,b_j)) = \ud(b_j,a(b_i,b_j))\\
		&= \min \{\ud(b_i,A), \ud(b_j,A)\}\\
		&= \min \{\ud(b'_i,A), \ud(b'_j,A)\}.\\
	\end{align*}
	Thus, $\{b'_j \mid b_j \in B(b_i)\} \subseteq B'(b'_i)$ for any $i\in \{1,\dots, l\}$.
	Analogously, we can argue that $\{b_j \mid b'_j \in B'(b'_i)\} \subseteq B(b_i)$.
	Therefore, $B'(b'_i)=\{b'_j \mid b_j \in B(b_i)\}$ and $B(b_i) = \{b_j \mid b'_j \in B'(b'_i)\}$.
	
	Furthermore, we show that $\ud(b_i,b_j) = \ud(b'_i,b'_j)$ if $b_j \not \in B(b_i)$ (equally, if $b_i \not \in B(b_j)$ or $b'_j \not \in B'(b'_i)$ or $b'_i \not \in B'(b'_j)$).
	To that end, let $b_j \not \in B(b_i)$.
	W.l.o.g., $\ud(b_i,A) \leq \ud(b_j,A)$ and, thus, $\ud(b_i,b_j) > \ud(b_i,A) = \ud(b_i,a)$ for some $a\in A$.
	Then, $\ud(b_i,b_j) = \ud(b_j,a)$ (recall for an ultrametric $\ud(\alpha,\beta)< \ud(\beta,\gamma)$ implies $\ud(\alpha,\gamma)=\ud(\beta,\gamma)$).
	Furthermore, $b'_j \not \in B'(b'_i)$ and thus, 
	\[
	\ud(b'_i,b'_j) > \min \{\ud(b'_i,A), \ud(b'_j,A)\} = \min \{\ud(b_i,A), \ud(b_j,A)\} = \ud(b_i,A) = \ud(b_i,a) = \ud(b'_i,a)\}.
	\]
	Therefore, also $\ud(b'_i,b'_j) =\ud(b'_j,a) = \ud(b_j,s) = \ud(b_i,b_j)$.
	
	Moreover, if $|B(b_i)|=|B'(b'_i)|=1$, then $\ud(b_i,A\cup B \setminus \{b_i\}) = \ud(b'_i,A \cup B' \setminus \{b'_i\}) \leq \ud(t'_i, T' -t'_i)$.
	
	Putting all this together, 
	\small
	\begin{align*}
		\sum_{i=1}^l &\ud(b_i,A\cup B \setminus \{b_i\}) \\
		&= \sum_{|B(b_i)|=1} (\ud(b_i, A\cup B \setminus \{b_i\}) + \ud(b_i,B \setminus \{b_i\}) - \ud(b_i,B \setminus \{b_i\}))+ \sum_{|B(b_i)| > 1 } \ud(b_i, A\cup B \setminus \{b_i\}) \\
		&= \sum_{|B(b_i)|=1} (\ud(b_i, A\cup B \setminus \{b_i\}) + \ud(b_i,B \setminus \{b_i\}) - \ud(b_i,B \setminus \{b_i\}))+ \sum_{|B(t_i)| > 1 } \ud(b_i, B\setminus \{b_i\}) \\
		&= \sum_{|B(b_i)|=1} (\ud(b_i, A\cup B \setminus \{b_i\}) - \ud(b_i,B \setminus \{b_i\}))+ \sum_{i=1}^l \ud(b_i, B\setminus \{b_i\}) \\
		&= \sum_{|B(b_i)|=1} (\ud(b_i, A\cup B \setminus \{b_i\}) - \ud(b_i,B \setminus \{b_i\}))+ \deltasummin(B) \\
		&\leq \sum_{|B'(b'_i)|=1} (\ud(b'_i, A\cup B' \setminus \{b'_i\}) - \ud(b'_i,B'\setminus \{b'_i\}))+ \deltasummin(B') \\
		&= \sum_{|B'(b'_i)|=1} (\ud(b'_i, A\cup B' \setminus \{b'_i\}) - \ud(b'_i,B' \setminus \{b'_i\}))+ \sum_{i=1}^l \ud(b'_i, B'\setminus \{b'_i\}) \\
		&= \sum_{|B'(b'_i)|=1} (\ud(b'_i, A\cup B' \setminus \{b'_i\}) + \ud(b'_i,B' \setminus \{b'_i\}) - \ud(b'_i,B' \setminus \{b'_i\}))+ \sum_{|B'(b'_i)| > 1 } \ud(b'_i, A\cup B' \setminus \{b'_i\}) \\
		&= \sum_i \ud(b'_i,A\cup B'\setminus\{b'_i\}).
	\end{align*}
	\normalsize
	In total, 
	\begin{align*}
		\deltasummin(A \cup B) &= \sum_{a\in A} \ud(a,A\cup B \setminus \{a\}) + \sum_{i=1}^l \ud(b_i,A\cup B \setminus \{b_i\})\\
		&\leq \sum_{a\in A} \ud(a,A\cup B' \setminus \{a\}) + \sum_{i=1}^l \ud(b'_i,A\cup B' \setminus \{b'_i\}) = \deltasummin(A\cup B').
	\end{align*}
	This concludes the proof.
\end{proof}

\subsection{Proof of Theorem~\ref{theo:hardness-sum-min}}

\begin{proof}
	The strategy is to reduce from \textsc{SAT}.
	To that end, we start by defining a fixed implicit representation $(\irI, \sem{\cdot})$ over some universe $\U$ which, together with a fixed ultrametric $\ud$, admits a polynomial time computable implicit ultrametric tree.
	The set of instances $\irI$ is the same as for \textsc{SAT}, i.e., propositional formulas $\phi(\bar{x})$ where, w.l.o.g, $\bar{x} = (x_1,\dots,x_n)$.
	The set $\U$ then consists of tuples $(\alpha,1), (\alpha,0)$ where $\alpha$ is an arbitrary truth assignment $\alpha\colon \bar{x} \rightarrow \{t,f\}$.
	Given a proposition formula $\phi(\bar{x})\in \irI$, the finite set of solutions $\sem{\phi(\bar{x})}\subseteq \U$ shall be $\sem{\phi(\bar{x})}=$
	\begin{align*}
		&\{(\alpha, 0) \mid \alpha \colon \bar{x} \rightarrow \{t,f\}\} \cup {}\\
		&\{(\alpha, 1) \mid \alpha \colon \bar{x} \rightarrow \{t,f\}, \alpha \vDash \phi\}.
	\end{align*}
	That is, $\sem{\phi(\bar{x})}$ consists of all truth assignments $\alpha$ (appended with a 0) and all models $\alpha$ of $\phi$ (appended with a 1).
	The distance between two elements $(\alpha,b),(\alpha', b')$ is $\ud((\alpha,b),(\alpha', b')) = 3^{-i}$ where $i$ is the last position $\alpha, \alpha'$ agree on, i.e., they differ first on $\alpha(x_{i+1}), \alpha'(x_{i+1})$ (or only one of $\alpha, \alpha'$ is defined on $x_{i+1}$ or they differ on $b,b'$).
	
	Note that there is a polynomial time computable implicit ultrametric tree $(\irroot, \irchildren, \irmember)$ for $\ud$ and $(\irI, \sem{\cdot})$.
	To that end, let $\phi(x_1,\dots,x_n)$ be an instance $\phi(x_1,\dots,x_n)\in \irI$ and let $S:=\sem{\phi(x_1,\dots,x_n)}$.
	Then, consider for each $\beta\colon \{x_1,\dots,x_i\} \rightarrow \{t,f\}$ where $i\in \{0,\dots, n\}$ the sets $B(\beta)=$
	\begin{align*}
		&\{(\alpha, 0) \mid \alpha \colon X \rightarrow \{t,f\}, \alpha(x_1)=\beta(x_1),\dots, \alpha(x_i)=\beta(x_i)\} \cup {}\\
		&\{(\alpha, 1) \mid \alpha \colon X \rightarrow \{t,f\}, \alpha(x_1)=\beta(x_1),\dots, \alpha(x_i)=\beta(x_i), \alpha \vDash \phi\}.
	\end{align*}
	The vertices of $\utree$ are exactly these balls $B(\beta)$ together with the singletons $\{(\alpha, 0)\},\{(\alpha, 1)\}$ for models $\alpha$ of $\phi(x_1,\dots,x_n)$.
	Thus, we can use $\beta, (\alpha, 0), (\alpha, 1)$ as IDs for our implicit ultrametric tree.
	
	Given this representation, $\irroot(\phi(x_1,\dots,x_n))$ simply outputs $B(\{\})=S$.
	We can compute the children of a ball $B(\beta)$ for partial truth assignments $\beta\colon \{x_1,\dots,x_i\} \rightarrow \{t,f\}, i\in \{0,\dots, n-1\}$ as $B(\beta\cup \{x_{i+1}\mapsto t\})$ and $B(\beta\cup \{x_{i+1}\mapsto f\})$.
	For total truth assignments $\beta\colon \{x_1,\dots,x_n\} \rightarrow \{t,f\}$, we have to check whether $\beta \vDash \phi$.
	If $\beta \not \vDash\phi$ then $B(\beta) = \{(\beta,0)\}$ and there are no children to enumerate while if $\beta\vDash \phi$ then $B(\beta)= \{(\beta, 0), (\beta, 1)\}$ and we enumerate both children $\{(\beta, 0)\}, \{(\beta, 1)\}$.
		
	Now, we proceed to the reduction.
	That is, we reduce a propositional $\phi(x_1,\dots, x_n)$ to the instance $\phi(x_1,\dots, x_n)\in \mathcal{\irI}$ itself but ask for a $(n+2)$-diverse subset $S'$ of $\sem{\phi(x_1,\dots,x_n)}$.
	
	We claim that such a $(n+2)$-diverse subset $S'$ has a diversity of at least $3^{-n} +\sum_{0\leq i \leq n} 3^{-i} = \frac{3+3^{-n}}{2}$ only if $S'$ contains a model of $\phi$ (recall that we are considering the diversity function $\deltasummin$).
	Moreover, if there exists a model of $\phi$, there also exists such a set $S'$.
	
	Let us first consider the ``if''-direction, i.e., let us assume there exists a model $\alpha$ of $\phi$.
	Then, let us define for $i\in \{1,\dots, n\}$
		\[
		\alpha_i (x_j) := \begin{cases}
			\alpha(x_j) & i\neq j\\
			\lnot \alpha & i = j
		\end{cases}
		\]
	and let us consider the set $S'=\{(\alpha_1,0),\dots, (\alpha_n,0), (\alpha,0), (\alpha,1)\}$.
	The diversity of $S'$ is $\deltasummin(S')=$
	\begin{align*}
		\deltasummin(S') &= d((\alpha,0), S' \setminus \{(\alpha,0)\}) + d((\alpha,1), S' \setminus \{(\alpha,1)\}) +\sum_{i=1}^n d((\alpha_i,0), S' \setminus \{(\alpha_i,0)\}) \\
		& = 3^{-n} + 3^{-n} +\sum_{i=1}^n 3^{i-1} = 3^{-n} + \sum_{i=0}^n 3^{i}.
	\end{align*}
	
	For the other direction, first consider a 2-diverse subset, which has a diversity of $2$.
	It has two elements $(\alpha_0,b_0), (\alpha_1,b_1)$ which already differ on $x_1$, i.e., both contribute a diversity of $1$.
	Then, a 3-diverse subset has a diversity of $1 + \frac{2}{3}$.
	It has the two previous elements together with an element $(\alpha_2,b_2)$ which, w.l.o.g., agrees with $\alpha_1$ on n $x_1$ but differs on $x_2$, i.e., $(\alpha_0,b_0)$ contributes a diversity of $1$ and $(\alpha_1,b_1)$ and $(\alpha_2,b_2)$ a diversity of $\frac{1}{3}$.
	Then, a 4-diverse subset has a diversity of $1 + \frac{1}{3} + \frac{1}{3}^2$.
	It has the three previous elements together with an element $(\alpha_3,b_3)$ which, w.l.o.g., agrees with $\alpha_2$ on $x_1$ and $x_2$ but differs on $x_3$, i.e., $(\alpha_0,b_0)$ contributes a diversity of $1$, $(\alpha_1,b_1)$ contributes a diversity of $\frac{1}{3}$, and both $(\alpha_2,b_2)$ and $(\alpha_3,b_3)$ a diversity of $\frac{1}{3}^2$.
	
	We can continue this until the size $|S| = n+1$ with diversity $\frac{3+3^{-n+1}}{2}$.
	The elements are $((\alpha_i,b_i))_{i=0,\dots,n}$ and $(\alpha_j,b_j),(\alpha_i,b_i), j<i$ agree on the variables $x_1,\dots,x_j$.
	Thus, $(\alpha_i,b_i)$ for $i< n$ contributes a diversity of $\frac{1}{3}^{i}$ and $(\alpha_{n},b_{n})$ also contributes a diversity of $\frac{1}{3}^{n-1}$.
	
	To continue the scheme to $|S| = n+2$ with diversity $\frac{3+3^{-n}}{2}$, the elements $(\alpha_{n},b_{n})$  and $(\alpha_{n+1},b_{n+1})$ now crucially have to agree on $\alpha_{n}, \alpha_{n+1}$ but differ on $b_{n}, b_{n+1}$.
	However, this is possible only if $(\alpha_{n}, 1) = (\alpha_{n+1}, 1)\in \sem{\phi}$, i.e., $\alpha_{n} = \alpha_{n+1}$ is a model of $\phi$.   
\end{proof}

\subsection*{Proof of Theorem~\ref{theo:fpt-implicit-rep}}

\begin{proof}
	\begin{algorithm}[t]
		\SetKwProg{Fn}{Function}{:}{}
		\KwIn{An instance $I \in \irI$ and $k \in \bbN$.}
		\KwOut{A $k$-diversity set $S' \subseteq \sem{I}$ with respect to $\delta$.}		
		$\mroot \gets \irroot(I)$ \\
		$S \leftarrow RelevantElements(I, \mroot, k)$\\
		\Return{$\argmax_{S'\subseteq S, |S'| = k} \delta(S')$}\\
		\Fn{$RelevantElements (I, B$, $k)$}{
			\If{$k=1$ or $|B| = 1$}{
				\Return$\{\irmember(I,B)\}$
			}
			$\irchildren(I,B).\irinit$ \\
			$C \leftarrow \{\texttt{Children}(I,B).\ircurrent\}$ \\
			\While{$\irchildren(I, B).\irnext = \texttt{true} \land |C| < k$ }{
				$C \leftarrow C \cup \{\texttt{Children}(I,B).\ircurrent\}$\\
			}
			$S \leftarrow \{\}$\\
			\For{$B_{\texttt{child}} \in C$} {
				$S \leftarrow S \cup RelevantElements(I, B_{\texttt{child}}, k-|C|+1)$
			}
			\KwRet{$S$}
		}
		\caption{For fixed ultrametric $\ud$ and implicit representation $(\irI, \sem{\cdot})$ over a common $\U$, implicit ultrametric tree $(\irroot, \irchildren, \irmember)$, and weakly monotone diversity function $\delta$ of~$\ud$, compute, for an instance $I \in \irI$, a $k$-diverse subset of $\sem{I}$ .}
		\label{alg:divAnsImpWM}
	\end{algorithm}

	In Algorithm~\ref{alg:divAnsImpWM}, we provide all the instructions for the algorithm of  Theorem~\ref{theo:fpt-implicit-rep}. This algorithm assumes a fixed ultrametric $\ud$ together with implicit representation $(\irI, \sem{\cdot})$ over a common universe $\U$.
	In addition, the algorithm makes use of an polynomial time computable implicit ultrametric tree given by the methods $(\irroot, \irchildren, \irmember)$ and a weakly monotone diversity function $\delta$ of~$\ud$. 
	Then, given an instance $I \in \irI$ and a $k \in \bbN$, the algorithm computes a $k$-diverse set $S' \subseteq \sem{I}$ with respect to $\delta$.
	
	We assume an iterator interface for $\irchildren$ as we did in the proof of Theorem~\ref{theo:ptime-implicit-rep}, i.e., with methods $\irinit$, $\irnext$, and $\ircurrent$. 
	Furthermore, $\irroot, \irchildren, \irmember$ all run in time $O(\fT(I))$ with $\fT(I)\leq |I|^\ell$ for some constant $\ell$.
	
	Intuitively, Algorithm~\ref{alg:divAnsImpWM} computes a set of ``relevant elements'' $S\subseteq \sem{I}$ such that there exists a $k$-diverse subset $S'$ of $\sem{I}$ which is also a subset of $S$.
	However, the size of $S$ is bounded by $2^k$ and, thus, computing a $k$-diverse subset of $S$ is possible in time $O(\binom{2^k}{k}\cdot\fdelta(k))$.
	
	To compute $S$, the algorithm proceeds recursively, starting with $\mroot = \sem{I}$ for which we are looking for elements such that the $k$ most diverse ones are among them (call in line 2).
	In the recursion, instead of $\mroot$ we could have any ball $B\in \Ball_\sem{I}$.
	
	Then, if $k=1$, it does not matter which element $a\in B$ is selected, or if $|B| = 1$, we can simply select the $a\in B$ as it is the only element we have at our disposal (lines 5,6).
	
	Otherwise, there are at least 2 children of $B$.
	To that end, let $B_1,\dots, B_l$ be the children of $B$.
	In that case, we recurse on the children $B_i$ with $i\in \{1,\dots, \min\{k,l\}\}$ and we are looking for elements of $B_i$ such that the $k- \min\{k,l\}+1$ most diverse ones are among them (the children as collected in lines 7-10 and the recursion happens in line 13).
	The reason behind this is that (due weak monotonicity) there exists a $k$-diverse subset $S'$ of $B$ such that $S'$ has at least 1 element from each of the children $B_i$ with $i\in \{1,\dots, \min\{k,l\}\}$.
	Intuitively, if $S'$ does not intersect some $B_i$, we can simply replace any $a\in S'\cap B_j$ from any $B_j\neq B_i$ with any $a'\in B_i$.
	
	The union of the elements deemed relevant for the children $B_i$ are then together the elements deemed relevant for $B$ (lines 11-14).
	
	\paragraph{Correctness} 
	We now prove the correctness of Algorithm~\ref{alg:divAnsImpWM}.
	To that end, we verify the following condition for any $I\in \irI, B\in \Ball_{\sem{I}}, k'\leq k$ by induction on $k'$:
	Let $S = RelevantElements(I, B, k')$. 
	For any $k$-subset $S'\subseteq \sem{I}$ with $|S'\cap B| =: m \leq  k'$, there exists a $m$-subset $A\subseteq S$ such that $\delta(S') \leq \delta(S'\setminus B \cup A)$ $(\ddagger)$.
	
	For $|B|=1$ we have $S = RelevantElements(I, B, k') = B$ and thus we can simply select $A := B$.
	
	For $k'=1$ we have $S = RelevantElements(I, B, k') = \{a\}$ where $a=\irmember(I,B)$.
	Then, let $S'\subseteq \sem{I}$ be as required.
	If $m=0$ we can again simply select $A := \emptyset$.
	Thus, assume $m=1$ and let $b$ be such that $S'\cap B = \{b\}$.
	Now consider $A = \{a\}$.
	For any $s\in S'\setminus B = S' \setminus \{b\}$ we have
	\[
	\ud(b,s) \leq \max\{\ud(b,a), \ud(a,s)\} = \ud(a,s)
	\]
	since $a,b$ are in the same ball $B\in \Ball_{\sem{I}}$ but $s$ is not in $B$.
	Thus, due to weak monotonicity, $\delta(S') \leq \delta(S'\setminus B \cup A)$
	
	Now consider a $I\in \irI, B\in \Ball_{\sem{I}}, 1<|B|, 1 < k'\leq k$ and assume $(\ddagger)$ holds for any $k''<k'$.
	Let $B_1,\dots, B_l$ be the children of $B$.
	Furthermore, let $S'\subseteq \sem{I}$ be as required.
	We define $S'_I:= S'\setminus B$ and $S'_B:= S'\cap B$.
	Let $B_i$ be a child of $B$ such that $B_i\cap S'_B = \emptyset$.
	Then, for any $a\in B_i, b\in S'_B$ and $s\in S'\setminus \{b\}$, again
	\[
	\ud(b,s) \leq \max\{\ud(b,a), \ud(a,s)\} = \ud(a,s).
	\]
	Thus, $\delta(S')\leq \delta(S'\setminus \{b\} \cup \{a\})$.
	It suffices to show Condition $(\ddagger)$ for $S'\setminus \{b\} \cup \{a\}$ (playing the role of $S'$) as this is strictly harder to achieve.
	Thus, we can require, w.l.o.g., $S'\cap B_1 \neq \emptyset, \dots, S'\cap B_{\min\{l,m\}} \neq \emptyset$.
	Furthermore, $|S'\cap B_i|\leq m-\min\{l,m\}+1\leq k'-\min\{l,k'\}+1<k'$ for all $i\leq \min\{l,m\}$.
	Thus, by the Condition $(\ddagger)$, there exist $(|S'\cap B_i|)$-subsets $A_i\subseteq RelevantElements(I, B_i, k'-\min\{l,k'\}+1)$ for which
	\[
	\delta(S'\setminus B \cup \bigcup_{j<i} A_j \cup \bigcup_{i\leq j}(S'\cap B_j)) \leq \delta(S'\setminus B \cup \bigcup_{j\leq i} A_j \cup \bigcup_{i< j}(S'\cap B_j)).
	\]
	Applying this from $i=1$ to $i=\min\{l,m\}$ and defining 
	\[
	A := \bigcup_{i=1}^{\min\{l,m\}} A_i\subseteq \bigcup_{i=1}^{\min\{l,m\}} RelevantElements(I, B_i, k'-\min\{l,k'\}+1) = RelevantElements(I, B, k')
	\] 
	this gives us
	\[
	\delta(S') = \delta(S'\setminus B \cup \bigcup_{1\leq j}(S'\cap B_j)) \leq \dots \leq \delta(S'\setminus B \cup \bigcup_{j\leq \min\{l,m\}} A_j) = \delta(S'\setminus B \cup A)
	\]
	as required.
	
	We can conclude that Condition $(\ddagger)$ holds for $B=\sem{I}$ and $k'=k$.
	Thus, for any $k$-subset $A\subseteq \sem{I}$, there exists a $k$-subset $S'\subseteq S:=RelevantElements(I, \sem{I}, k)$ such that $\delta(A)\leq \delta(S')$.
	Consequently, $S$ contains at least one $k$-diverse subset of $\sem{I}$.
	
	\paragraph{Running time} 
	By considering the recursion tree Algorithm~\ref{alg:divAnsImpWM} implicitly traverses, one can check that the number of balls that are used from the the ultrametric tree is at most $O(2^k)$ (the worst that can happen is that every ball has 2 children and we have to explore the whole (binary) ultrametric tree till depth $k$). 
	Thus, the recursion takes time at most $O(2^k \cdot |I|^\ell)$.
	
	Lastly, computing the $\argmax$ in line 3 may take up to $O(\binom{2^k}{k}\cdot \fdelta(k))$ time where $\fdelta(k)$ represents the time required to compute $\delta$ over a set of $k$ elements.´
	Note that all of this is fixed parameter tractable.
\end{proof}
 	
	\section{Proofs of Section~\ref{sec:acq}} \label{app:acq}

\subsection*{Proof of Theorem~\ref{cor:acq-evaluation}}

\begin{proof}
	We want to apply Theorem~\ref{theo:ptime-implicit-rep} and, therefore, we need to give an implicit ultrametric tree $(\irroot, \irchildren, \irmember)$ which runs in time $O(|Q| \cdot |D|\cdot \log(|D|))$.
	To that end, we first define IDs for the balls $\Ball_{\sem{Q}(D)}$ which will be subsequently used by the implicit ultrametric tree.
	Concretely, for every ball $B \in \Ball_{\sem{Q}(D)}$, there exists values  $c_1, \ldots, c_i$ such that $B = \{Q(\bar{a}) \in \sem{Q}(D) \mid \forall j \leq i. \,\bar{a}[j] = c_j\}$.
	This means that all answers in $B$ agree on the values  $c_1, \ldots, c_i$ and these form a common prefix.
	Moreover, these prefixes are as long as possible, i.e., $i$ is as big as possible, and we can use $c_1, \ldots, c_i$ to uniquely identify $B$.
	
	To implement the methods $\irroot$, $\irchildren$, and $\irmember$, we modify Yannakakis algorithm~\cite{yannakakis1981algorithms}.
	Recall that Yannakakis algorithm proceeds in the following manner (we only sketch the preprocessing phase) on an ACQ $Q(\bar{x}) \leftarrow R_1(\bx_1), \ldots, R_{m}(\bx_{m})$:
	\begin{enumerate}
		\item The $R_i$ are arranged in a tree structure (in a join tree).
		\item Each $R_i$ gets assigned a unique copy $R^D_i$ of the corresponding table in $D$.
		\item The $R^D_i$ are semijoined as to delete dangling tuples.
	\end{enumerate}
	Thus, after the preprocessing phase, $R_i(h(\bx_{i}))\in R^D_i$ for some $h_i \colon \bx_{i} \rightarrow \D$ if and only if $h_i$ can be extended to a $h\colon \Var \rightarrow \D$ such that $Q(h(\bx))\in \sem{Q}(D)$.
	Thus, the admissible values 
	\[
	ad(x):= \{d\in \D \mid \exists h\colon \Var\rightarrow\D \text{ s.t., } h(x)=d \text{ and } Q(h(\bx))\in \sem{Q}(D)\}
	\]
	of a variable $x\in \Var$ can be compute in time $O(|D|)$ by inspecting a table $R^D_i$ where $x$ appears in $\bx_i$.
	Therefore, to compute the common prefix of $\mroot:=\sem{Q}(D)$, we simply have to iteratively go through $\bx[1], \dots, \bx[|\bx|]$ to find the first $\bx[i]$ such that $|ad(\bx[i])|\neq 1$.
	The common prefix of $\mroot$ then is $c_1, \ldots, c_{i-1}$ where $\{c_1\} = ad(\bx[1]), \dots, \{c_{i-1}\} = ad(\bx[i-1])$ (for $i=1$ the common prefix is the empty prefix $\epsilon$).
	Computing this prefix takes time $O(|Q|\cdot |D| \cdot \log(|D|))$ and is exactly what $\irroot$ does.
	
	To compute the children of a ball $B\in \Ball_{\sem{Q}(D)}$ with common prefix $\bar{c} = (c_1,\dots, c_{i-1})$ we can do the following:
	Consider the query $Q'(\bar{x}[i],\dots,\bar{x}[|\bar{x}|]) \leftarrow R_1(h(\bx_1)), \ldots, R_{m}(h(\bx_{m}))$ where $h$ is an partial assignment that maps $h(\bar{x}[j]) = c_j$ for every $j \leq i-1$ and $h(x) = x$ for any other variable $x \in \Var$. 
	I.e., we plugged the prefix $\bar{c}$ into the query $Q$.
	We can then compute the admissible values of the next variable $\bx[i]$ for the prefix $\bar{c}$, i.e., the set
	\[
	ad_{\bar{c}}(\bx[i]):= \{d\in \D \mid \exists h\colon \Var\rightarrow\D \text{ s.t., } h(\bx[i])=d \text{ and } Q'(h(\bar{x}[i],\dots,\bar{x}[|\bar{x}|]))\in \sem{Q'}(D)\}.
	\]
	This may take time up to $O(|Q|\cdot |D| \cdot \log(|D|))$.
	Note that $|ad_{\bar{c}}(\bx[i])|\geq 1$ as otherwise the prefix of $B$ would have been longer.
	Then, to enumerate the children we iterate through $c_i\in ad_{\bar{c}}(\bx[i])$.
	For a $c_i$, we plug in the new prefix $(\bar{c},c_i)$ into $Q$ which results in the query $Q'_{c_i}(\bar{x}[i+1],\dots,\bar{x}[|\bar{x}|]) \leftarrow R_1(h(\bx_1)), \ldots, R_{m}(h(\bx_{m}))$.
	Then, let $B_{c_i}$ be the answers $\sem{Q'_{c_i}}(D)$ prepended by the prefix $(\bar{c},c_i)$.
	Note that $B_{c_i}$ is a child of $B$ and we can compute the prefix of it as the prefix $\bar{c}'_{c_i}$ of $\sem{Q'_{c_i}}(D)$ prepended by the prefix $(\bar{c},c_i)$, i.e., it is $(\bar{c},c_i, \bar{c}'_{c_i})$ ($\bar{c}'_{c_i}$ may be the empty prefix).
	All of this takes time $O(|Q|\cdot |D| \cdot \log(|D|))$ for each $c_i$ and, thus, this is also the delay we get for $\irchildren$.
	
	Lastly, given a ball $B\in \Ball_{\sem{Q}(D)}$ with common prefix $\bar{c} = (c_1,\dots, c_{i-1})$ we can easily compute a $b\in B$ by computing any answer $a\in \sem{Q'}(D)$ with Yannakakis algorithm and prepend it with $\bar{c}$.
	Thus, $\irmember$ also only requires time $O(|Q|\cdot |D| \cdot \log(|D|))$.
\end{proof}

\subsection*{Proof of Theorem~\ref{theo:disruptive-trio}}

\begin{proof}
	We proceed similar to the proof of Theorem~\ref{cor:acq-evaluation} but in the absence of a disruptive trio we can move to an extension of Yannakakis algorithm developed in~\cite{DBLP:journals/tods/CarmeliTGKR23} which takes the order of the head variables $\bx$ in consideration.
	Intuitively, the absence of a disruptive trio ensures the existence of a \emph{layered join tree} whose layers follow the order of the variables in the head of $Q$.
	This will allow us to improve the runtime of $\irroot, \irchildren$ and $\irmember$ to $O(|Q|)$ if we allow a common preprocessing of $O(|Q|\cdot |D|\cdot \log(|D|))$.
	Thus, be inspecting how we arrive at the run time in Theorem~\ref{theo:ptime-implicit-rep} this then totals to the run time as required.
	
	We start by recalling the steps taken by the algorithm presented in~\cite{DBLP:journals/tods/CarmeliTGKR23} (we only sketch the preprocessing phase) on a free-connex ACQ $Q(\bar{x}) \leftarrow R_1(\bx_1), \ldots, R_{m}(\bx_{m})$ without a disruptive trio:
	\begin{enumerate}
		\item Projections of some $R_i(\bx_i)$ are possibly added to the query such that the resulting query has a layered join tree.
		However, the semantics of the query remains unchanged and, thus, we assume that $Q$ already includes all of these projections needed.
		\item The $R_i$ are arranged in a rooted tree structure $T$ (a layered join tree).
		As $Q$ is free-connex we assume that the free variables $\bx$ appear in a subtree which includes the root.
		\item Each $R_i$ gets assigned a unique copy $R^D_i$ of the corresponding table in $D$.
		\item The $R^D_i$ are semijoined as to delete dangling tuples.
		Furthermore, for each $R_i$ with child $R_j$, indexes are created such that we can access the join partners $t_j \in R^D_j$ of each $t_i\in R^D_i$ in constant time and with constant delay.
		\item As $Q$ is free-connex, we can now remove the bounded variables.
		Thus, w.l.o.g., we may assume $Q$ to be a full CQ, i.e., all $\bx_i$ are all sequences of variables in $\bx$.
	\end{enumerate}
	All of this preprocessing only requires $O(|Q|\cdot |D|\cdot \log(|D|))$ time.
	Furthermore, due to the layered join tree we can assume, w.l.o.g., that $R_i$ is the root of the subtree of $T$ containing the variable $\bx[i]$ and if $R_i$ is the parent of $R_j$ that $i<j$.
	Also, to simplify the subsequent presentation we assume the variables in $\bx_i$ to be ordered according to $\bx$ and that variables are not repeated within $\bx_i$ nor $\bx$.
	
	Now, we can proceed to define the algorithms $(\irroot, \irchildren, \irmember)$ similar to the proof of Theorem~\ref{cor:acq-evaluation} but it will no longer be necessary to recompute join trees and we can always remain in $T$.
	To that end, recall the definition of the prefix of a ball $B\in \Ball_{\sem{Q}(D)}$ as the values $c_1, \ldots, c_i$ such that $B = \{Q(\bar{a}) \in \sem{Q}(D) \mid \forall j \leq i. \,\bar{a}[j] = c_j\}$, and the admissible values for a variable $x\in \Var$, i.e.,
	\[
	ad(x):= \{d\in \D \mid \exists h\colon \Var\rightarrow\D \text{ s.t., } h(x)=d \text{ and } Q(h(\bx))\in \sem{Q}(D)\}.
	\]
	Furthermore, we also define the admissible values of a variable $x\in \Var$ given a prefix $\bar{c}=(c_1,\dots, c_{i-1})$.
	Slightly different to before, we define
	\[
	ad_{\bar{c}}(\bx[i]):= \{d\in \D \mid \exists h\colon \Var\rightarrow\D \text{ s.t., } h((\bx[1],\dots,\bx[i-1]))=\bar{c}, h(x[i])=d \text{ and } Q(h(\bar{x}))\in \sem{Q}(D)\}.
	\]
	
	Note that computing $ad(\bx[1])$ is easy as $\bx_1$ can only contain the variable $\bx[1]$ due to the fact that $R_1$ is the root of $T$ but only the root of the subtree containing the variable $\bx[1]$.
	Thus, $R_1(\bx[1])$ is part of the query and $ad(\bx[1]) = R^D_1$.
	If $|R^D_1| = |ad(\bx[1])|>1$, the prefix of $\mroot:=\sem{Q}(D)$ is the empty prefix $\epsilon$ and there is nothing more to do for $\irroot$.
	Otherwise, we proceed to $ad(\bx[2])$.
	
	To that end, let in general $\{c_1\} = ad(\bx[1]), \dots, \{c_{i-1}\} = ad(\bx[i-1])$ and we are looking at whether $ad(\bx[i])$ is of size 1 or greater than 1.
	To that end, consider $R_i$ which is the root of the subtree of $T$ containing the variable $\bx[i]$.
	Hence, in $Q$, it may only appear together with variables $\bx[j]$ with $j\leq i$.
	But we know all of them only have 1 admissible value, thus, $R^D_i$ is the same as $ad(\bx[i])$ where values $c_j$ for the $x[j]$ which appear in $\bx_i$ and where $j<i$ are prepended to $ad(\bx[i])$.
	I.e., for $h\colon \{\bx[1], \dots, \bx[i-1]\} \rightarrow \D, h(\bx[l]) := c_l$ 
	\[
	R^D_i = \{(h(\bx_i \setminus \{\bx[i]\}), d) \mid d\in ad(\bx[i])\}.
	\]
	Thus, if $|R^D_i| = |ad(\bx[i])|>1$, the prefix of $\mroot:=\sem{Q}(D)$ is $(c_1,\dots, c_{i-1})$ and there is nothing more to do for $\irroot$.
	Otherwise, we proceed to $ad(\bx[i+1])$.
	
	In total, applying the preprocessing as sketched above, $\irroot$ only requires $O(|Q|)$ time.
	
	Now lets proceed to computing the children of a ball $B\in \Ball_{\sem{Q}(D)}$.
	To that end, let $\bar{c}=(c_1,\dots, c_{i-1})$ be the common prefix of $B$.
	We go to $R_i$ and its parent $R_j$ where $j<i$.
	Similar to before, we can determine $ad_{\bar{c}}(\bx[i])$ by inspecting $R_i$.
	To that end, let $h\colon \{\bx[1], \dots, \bx[i-1]\} \rightarrow \D, h(\bx[l]) := c_l$ .
	Then,
	\[
	\{ (h(\bx_i \setminus \{\bx[i]\}), d) \in R^D_i \} = \{(h(\bx_i \setminus \{\bx[i]\}), d) \mid d\in ad_{\bar{c}}(\bx[i])\}.
	\]
	The left hand side are the tuples of $R^D_i$ that adhere to the prefix $\bar{c}$ while the right hand side are the admissible values of $\bx[i]$ given the prefix $\bar{c}$ prepended by the values of $\bar{c}$ that correspond to variables appearing in $\bx_i$.
	However, given the indexes on the parent $R_j$, we can compute the left hand side with constant delay.
	To see let consider $t_j:= h(\bx_j) \in R^D_j$ and notice that $\{ (h(\bx_i \setminus \{\bx[i]\}), d) \in R^D_i \}$ are exactly the join partners of $t_j$ in $R^D_i$.
	Thus, we can iterate through $c_i\in ad_{\bar{c}}(\bx[i])$ with constant delay.
	
	Now let $B_{c_i}\in \Ball_{\sem{Q}(D)}$ be the answers with prefixes $(\bar{c},c_i)$.
	Note that $B_{c_i}$ is a child of $B$ but we still have to extend this prefix to the maximal prefix for $B_{c_i}$.
	To that end, we have to inspect $ad_{(\bar{c},c_i)}(\bx[i+1])$.
	However, we already know how to compute this by following the same argumentation as for $ad_{\bar{c}}(\bx[i])$.
	If $|ad_{(\bar{c},c_i)}(\bx[i+1])|>1$ we stop and assert that $(\bar{c},c_i)$ is the correct maximal prefix.
	Otherwise, we continue to a $l>1$ such that $\{c_{i+1}\}= ad_{(\bar{c},c_i)}(\bx[i+1]), \dots, \{c_{i+l-1}\} = ad_{(\bar{c},c_i)}(\bx[i+l-1])$ and $|ad_{(\bar{c},c_i)}(\bx[i+l])|>1$.
	Then, $(\bar{c},c_i, \dots, c_{i+l-1})$ is the maximal prefix of $B_{c_i}$.
	In total, $\irchildren$ has a delay of at most $O(|Q|)$.
	
	Lastly, given a ball $B\in \Ball_{\sem{Q}(D)}$ with common prefix $\bar{c} = (c_1,\dots, c_{i-1})$ we can easily compute a $b\in B$.
	To do that, we iterate through $R_i,\dots, R_m$ in this order.
	Let $h\colon \{\bx[1], \dots, \bx[i-1]\} \rightarrow \D, h(\bx[l]) := c_l$.
	We process $R_i$ with parent $R_j$ by considering the tuple $t_j:=h(\bx_j)\in R^D_j$ and simply select the first join partner $t_i\in R_i$.
	Then, $t_i$ assigns a value to $\bx[i]$ which we call $c_i$.
	Now we can proceed to $R_{i+1}$ with the prefix $(c_1,\dots, c_{i})$, i.e., the previous prefix extended by $c_i$.
	Also this process, i.e., $\irmember$, only requires $O(|Q|)$ time.
		
	By reinspecting how we arrive at the run time in Theorem~\ref{theo:ptime-implicit-rep} -- in particular what role the run times of $(\irroot, \irchildren, \irmember)$ play -- this then totals to the run time as required.
\end{proof}

\subsection*{Proof of Theorem~\ref{theo:disruptive-trio-W}}

\begin{proof}
	Let us reconsider the proof of Theorem~\ref{theo:disruptive-trio} and the definitions used there.
	Furthermore, let us consider the execution of Algorithm \ref{alg:divAnsImpSSM} using the implicit ultrametric tree developed in the proof of Theorem~\ref{theo:disruptive-trio}.
	To that end, let $S, L$ be as they are at the start of some loop iteration, i.e., in line 6.
	Moreover, let $L=\{B_1,\dots, B_l\}$ and let $\bar{p}_i$ be the prefix corresponding to the ball $B_i\in \Ball_{\sem{Q}(D)}$ and $\bar{c}_i$ be the prefix corresponding to the current children, i.e., of the ball $\irchildren(I,B_i).\ircurrent$.
	This means that there is an answer $h\in S \subseteq \sem{Q}(D)$ with the prefix $p_i$ for any $i=1,\dots,l$ but there is no answer with the prefix $\bar{c}_{i}[1],\dots, \bar{c}_{i}[|\bar{p}_i +1|]$.
	Therefore, the incremental diversity of each $b_i:=\irmember(I,\irchildren(I,B_i).\ircurrent)$ is
	\[
	\delta(S\cup \{b_i\})-\delta(S) = \udR(b_i,S) = 2^{-|\bar{p}_i|-1}.
	\]
	Thus, by storing $L$ as an array (of length $|Q|$) of sets with $B_i$ in the set at position $|\bar{p}_i +1|$, we can find a $B$ as required in line 7 in time $|Q|$.
		
	By reinspecting how we arrive at the run time in Theorem~\ref{theo:ptime-implicit-rep} -- in particular why the term $k^2\cdot f(k)$ arises -- this then totals to the run time as required.
\end{proof}

\end{document}